\documentclass[12pt, a4paper]{article}

\usepackage[utf8]{inputenc}
\usepackage[T1]{fontenc}
\usepackage{setspace}
\usepackage{amsmath, amssymb, amsfonts, amsthm}
\usepackage{mathtools}
\usepackage{enumitem}
\usepackage{booktabs}
\usepackage{tabularx}
\usepackage{threeparttable}
\usepackage{graphicx}
\usepackage{float}
\usepackage{hyperref}
\usepackage{natbib}
\usepackage{caption}

\usepackage{booktabs}
\usepackage{tabularx}
\usepackage{array}

\newcolumntype{L}[1]{>{\raggedright\arraybackslash}p{#1}}
\newcolumntype{Y}{>{\raggedright\arraybackslash}X}

\newtheorem{proposition}{Proposition}
\newtheorem{lemma}{Lemma}

\theoremstyle{definition}
\newtheorem{definition}{Definition}
\newtheorem{assumption}{Assumption}

\theoremstyle{remark}

\usepackage[utf8]{inputenc}
\usepackage[T1]{fontenc}
\usepackage{amsmath,amssymb,amsthm}
\usepackage{graphicx}
\usepackage{booktabs}
\usepackage{tabularx}
\usepackage{multirow}
\usepackage[margin=1in]{geometry}
\usepackage{xcolor}

\hypersetup{colorlinks=true,linkcolor=blue!60!black,citecolor=blue!60!black,urlcolor=blue!60!black}

\title{
Delegation Rights: Property, Agency, and Investment Incentives\\
in the Age of AI Agents
}

\author{
Yukun Zhang\\
The Chinese University of Hong Kong\\
Hong Kong, China\\
\texttt{215010026@link.cuhk.edu.cn}
\and
Kemu Xu\\
University of Edinburgh\\
Edinburgh, United Kingdom\\
\texttt{s2749200@ed.ac.uk}
}

\begin{document}

\maketitle


\begin{abstract}
\noindent
AI agents increasingly operate inside digital accounts rather than merely around platform interfaces. A user may authorize an agent to search, compare, draft, book, or transact by exercising privileges that the user already holds. This paper studies who should control this mode of account use: the platform, the user, or a conditional certification regime. We define \emph{delegation rights} as the revocable, identity-preserving, scope-limited, and mode-specific authority of an account holder to authorize an automated proxy to exercise existing account privileges on her behalf.

We develop a three-party incomplete-contracts model with a User, an AI Agent provider, and a Platform. The contested control object is not platform ownership, account transferability, data portability, or unrestricted API access, but the residual authority to determine whether an existing account entitlement must be exercised manually or may be exercised through a user-authorized automated proxy. Under Platform Control, the platform can protect infrastructure, identity systems, privacy boundaries, and third parties, but its discretionary veto over automated access weakens the User--Agent coalition's disagreement payoff and depresses relationship-specific investment. Under User Control, the hold-up problem is reduced, but some security, privacy, congestion, and third-party risks may remain outside the agent's private incentives.

We then analyze \emph{Certified Delegation}, under which access protection is conditional on verifiable requirements such as explicit authorization, revocability, auditability, rate-limit compliance, data minimization, and risk mitigation. Certification is therefore not merely a technical safety screen; it is a conditional allocation of residual control. A certified agent receives a protected access path, while an uncertified or non-compliant agent remains subject to platform refusal. Illustrative mechanism simulations, not intended as structural estimates of any specific dispute, show how such a conditional regime can reduce deadweight loss relative to the two polar regimes by restoring productive delegation incentives while bounding residual risk.


\end{abstract}

\newpage
\tableofcontents
\newpage

\section{Introduction}
\label{sec:introduction}

\subsection{Motivation: From Account Use to Delegated Account Operation}
\label{subsec:intro_motivation}

Digital platforms are organized around user accounts. A user logs in, searches, communicates, purchases, pays, books, or manages information inside an environment governed by the platform's technical design and contractual rules. For much of the internet's history, this arrangement rested on a relatively stable premise: the account holder was also the person operating the account. External intermediaries existed, but they usually remained outside the private account boundary. Search engines indexed public information; browser extensions modified local display; and platform-approved APIs allowed developers to interact only through interfaces specified by the platform.

AI agents make this premise less stable. A user may now authorize a software agent to compare products, complete forms, summarize messages, draft replies, coordinate bookings, or execute multi-step workflows across digital services. The agent does not own the account, and it does not ordinarily claim an independent entitlement against the platform. It acts for a user who already has permission to perform the relevant actions manually. This raises a narrower question than general platform access: if a user may do something inside her account, may she choose an automated proxy to do it on her behalf?

The question concerns the \emph{mode} through which an existing account entitlement is exercised. Platforms have legitimate reasons to regulate that mode. Automated operation may increase system load, weaken fraud-detection signals, bypass advertising or ranking interfaces, expose third-party data, or blur the boundary between human and machine communication. At the same time, a broad platform veto over user-authorized automated proxies can protect incumbent business models, limit user-side intermediation, and discourage investment in valuable automation. The same technical restriction can therefore serve both safety-preserving and exclusionary functions.

This makes account-level delegation a property-rights problem. The platform's ability to classify user-authorized machine operation as impermissible access functions as a residual control right over the user's mode of account execution. The contested asset is not the platform's infrastructure, data, or user account itself. It is the residual authority to determine whether an existing account entitlement must be exercised manually or may be exercised through a user-authorized automated proxy.

The allocation of this authority affects investment. Users may invest in workflow design, preference specification, and monitoring routines. Agent providers may invest in platform-specific adaptation, execution reliability, authorization controls, and compliance systems. Platforms may invest in authentication, privacy protection, payment security, fraud detection, and infrastructure stability. These investments are relationship-specific, while many relevant contingencies are difficult to specify in advance. Agent capabilities change, platform interfaces are redesigned, security vulnerabilities emerge, and regulatory expectations evolve. The party that controls whether machine-mediated account use is permitted therefore also controls an important bargaining position after investments are sunk.

This paper studies that control problem. The argument is not that users should have an unrestricted right to automate all platform interactions. Nor is it that platforms should have an unconditional right to exclude every machine proxy. The narrower claim is that AI agents create a distinct margin of digital control: the choice of whether an account holder may exercise existing rights through an authorized automated representative. The efficient allocation of this margin depends on both investment incentives and safety externalities.

\subsection{Delegation Rights and Research Question}
\label{subsec:intro_question}

We call this control margin a \emph{delegation right}. A delegation right is the revocable, identity-preserving, scope-limited, and mode-specific authority of an account holder to authorize an automated proxy to exercise existing account privileges on her behalf.

The qualifications are central to the concept. The right is \emph{revocable} because the user must be able to terminate proxy access. It is \emph{identity-preserving} because the agent acts as the user's representative rather than as a new account owner. It is \emph{scope-limited} because the agent cannot exceed the user's existing contractual and technical permissions. It is \emph{mode-specific} because it concerns how an existing right is exercised, not whether the underlying right exists.

Delegation rights are therefore distinct from neighboring legal and economic categories. They are not platform ownership, because they do not transfer control over servers, code, databases, matching systems, identity systems, or governance architecture. They are not ordinary use rights, which determine whether a user may participate in the platform. They are not account transferability, which changes the identity of the entitlement holder. They are not data portability, which concerns the extraction and movement of data across services. They are also not equivalent to API access, which is usually defined by platform-controlled endpoints and developer agreements. Delegation rights concern a different question: whether the user may select an automated proxy as the operational means of using rights she already holds.

The paper asks:

\begin{quote}
\emph{When a user authorizes an AI agent to operate within a digital platform account, who should hold residual control over that delegated operation: the platform, the user, or a conditional certification regime?}
\end{quote}

This question differs from standard platform-access disputes. In a conventional complementor-access case, an external developer or rival seeks entry into a platform ecosystem on its own behalf. In the delegation setting, the agent's authority is derivative of the user. The relevant institutional object is not entry as an independent platform participant, but proxy execution of pre-existing account privileges. That distinction changes both the property-rights analysis and the design of possible remedies.

\subsection{Core Argument}
\label{subsec:intro_argument}

The analysis compares three regimes.

Under \emph{Platform Control}, the platform may exclude automated proxies at its discretion. This regime helps the platform protect infrastructure, identity systems, privacy boundaries, payment integrity, fraud-detection systems, and network stability. Its cost is hold-up. Once the user and the agent provider have invested in platform-specific workflows, the platform can threaten exclusion or extract surplus. Anticipating this, the User--Agent coalition underinvests in automation, adaptation, and compliance.

Under \emph{User Control}, the account holder has a protected right to delegate account operation to an automated proxy. This strengthens the outside option of the User--Agent coalition and reduces platform hold-up. But it creates a different distortion. The agent provider does not necessarily bear all of the costs created by automated execution. Some risks fall on the platform, non-delegating users, counterparties, or the broader network. Unconditional delegation may therefore lead to insufficient safety investment, excessive system load, privacy leakage, or weakened accountability.

The third regime is \emph{Certified Delegation}. Under this arrangement, delegation is protected only when the agent satisfies verifiable requirements for authorization, auditability, rate-limit compliance, data minimization, security, and accountability. A certified agent receives a credible access path. The platform cannot refuse access merely because the operator is automated. An uncertified or non-compliant agent remains subject to exclusion.

Certification is not only a safety screen; it is a conditional allocation of residual control. It changes the platform's right of refusal. When the agent satisfies the relevant standard, the platform's discretion to exclude is limited. When the agent fails the standard, the platform retains authority to protect its system and affected third parties. Certified Delegation is therefore a second-best arrangement: it limits arbitrary exclusion without eliminating the platform's ability to refuse unsafe or non-compliant automation.

\subsection{Scope Conditions}
\label{subsec:intro_scope}

The argument is conditional rather than absolutist. Delegation rights do not protect automation that exceeds the user's account privileges, bypasses security systems, violates revocation, ignores rate limits, or exposes third-party data outside the authorized task scope. They do not create a general right to scrape, circumvent authentication, defeat CAPTCHA systems, evade payment security, or override access controls. They also do not give third-party agent providers an independent entitlement to enter a platform without user authorization.

Nor do delegation rights require platforms to surrender ownership of their infrastructure or expose unrestricted APIs. A delegation right concerns the permissible mode of exercising an existing account entitlement. It does not expand the substantive bundle of account rights, transfer the user's account to the agent, or require the platform to redesign its entire technical architecture for external automation.

These scope conditions are important for both the theory and the policy analysis. If automation is unauthorized, non-revocable, out of scope, non-auditable, or unsafe, the delegation claim fails. The relevant question is not whether automation should always be allowed. It is whether authorized, bounded, revocable, and accountable proxy execution should receive protection against arbitrary exclusion.

\subsection{Model and Main Results}
\label{subsec:intro_results}

We study the problem in a three-party incomplete-contracts model with a User ($U$), an AI Agent provider ($A$), and a Platform ($P$). Before bargaining, the parties choose non-contractible, relationship-specific investments. The user invests in workflow design, preference specification, and monitoring. The agent provider invests in platform-specific execution capability and compliance-related adaptation. The platform invests in infrastructure, authentication, fraud prevention, and system stability.

Automated account use generates trilateral surplus but may also impose a safety externality. Ex-post surplus is divided through Shapley bargaining. The governance regime changes coalition values and disagreement payoffs, which then feed back into ex-ante investment incentives. Platform Control, User Control, and Certified Delegation differ not because they change the underlying production technology, but because they allocate residual control over delegated account execution differently.

The analysis produces five main results.

First, Platform Control depresses user and agent investment. Because the platform can block automated access after investments are sunk, the User--Agent coalition has a weak disagreement payoff. This lowers the private marginal return to workflow design and platform-specific agent adaptation, generating underinvestment relative to a benchmark in which relationship-specific investments are fully rewarded.

Second, User Control mitigates hold-up but does not solve the safety problem. A protected delegation right improves the User--Agent coalition's outside option. However, the agent captures only part of the social value of risk reduction. Its private incentive to invest in safety and compliance is therefore below the social incentive when some harms are borne by the platform or third parties.

Third, Certified Delegation creates a conditional allocation of control. If the agent satisfies a verifiable risk standard, the effective regime resembles User Control. If it fails, the regime reverts to Platform Control. This protects qualified agents from arbitrary exclusion while preserving the platform's right to block unsafe proxies.

Fourth, the preferred regime depends on the environment. User delegation is more attractive when user adaptation and agent capability are central to value creation and when non-cooperative workarounds are effective. Platform control is more attractive when infrastructure integrity, identity protection, social-graph privacy, or third-party externalities dominate the welfare calculation. The model therefore supports a contribution-threshold logic rather than a universal rule.

Fifth, illustrative mechanism simulations show how the model behaves under maintained parameter values. These simulations are not intended as structural estimates of any particular platform dispute. They show that Platform Control can generate hold-up, User Control can leave residual externalities, and Certified Delegation can reduce deadweight loss by restoring productive delegation incentives while bounding residual risk.

\subsection{Contributions}
\label{subsec:intro_contributions}

The paper makes three contributions.

First, it makes a conceptual contribution by identifying delegation rights as a distinct margin in the digital property-rights bundle. The relevant object of control is not platform infrastructure, account ownership, data portability, account transferability, ordinary use, or API access. It is the operational mode through which an account holder exercises existing rights. This distinction matters because user-authorized AI agents occupy an intermediate institutional position: they are neither ordinary human users nor independent platform complementors.

Second, it makes a theoretical contribution by modeling AI-agent account operation as a residual-control allocation problem. In the model, the account interface becomes the relevant control point. A platform veto over automated operation weakens the User--Agent coalition's disagreement payoff and reduces relationship-specific investment. A protected delegation right improves the coalition's outside option but may leave some safety and privacy costs insufficiently internalized. The model therefore links platform control, user-agent bargaining power, investment incentives, and welfare in a single incomplete-contracts framework.

Third, it makes an institutional contribution by interpreting certification as a conditional allocation of residual control. In many policy discussions, certification is treated mainly as a technical safety screen. Here, certification also changes bargaining power and access rights. A certified proxy receives a protected access path; an uncertified or non-compliant proxy remains excludable. This interpretation connects AI-agent safety, platform governance, and property-rights theory.

These contributions are intentionally limited. The paper does not defend a general right to automate. It defends a qualified delegation right: user-authorized, revocable, scope-limited, identity-preserving, auditable, and proportionate to risk. The policy implication is not unconditional openness, but certification-contingent access protection.

\subsection{Roadmap}
\label{subsec:intro_roadmap}

The rest of the paper proceeds as follows. Section~\ref{sec:institutional} develops the institutional foundation of delegation rights and distinguishes them from ownership, ordinary use, account transferability, data portability, and API access. Section~\ref{sec:lit} relates the argument to property-rights theory, platform governance, interoperability mandates, cyberlaw, and AI-agent safety certification. Section~\ref{sec:model} presents the incomplete-contracts model and compares Platform Control with User Control. Section~\ref{sec:certified} introduces Certified Delegation and derives the relevant incentive and participation constraints. Section~\ref{sec:calibration} provides illustrative mechanism simulations and stylized case counterfactuals. Section~\ref{sec:policy_conclusion} discusses policy implications and concludes.

\begin{figure}[htbp]
\centering
\includegraphics[width=\textwidth]{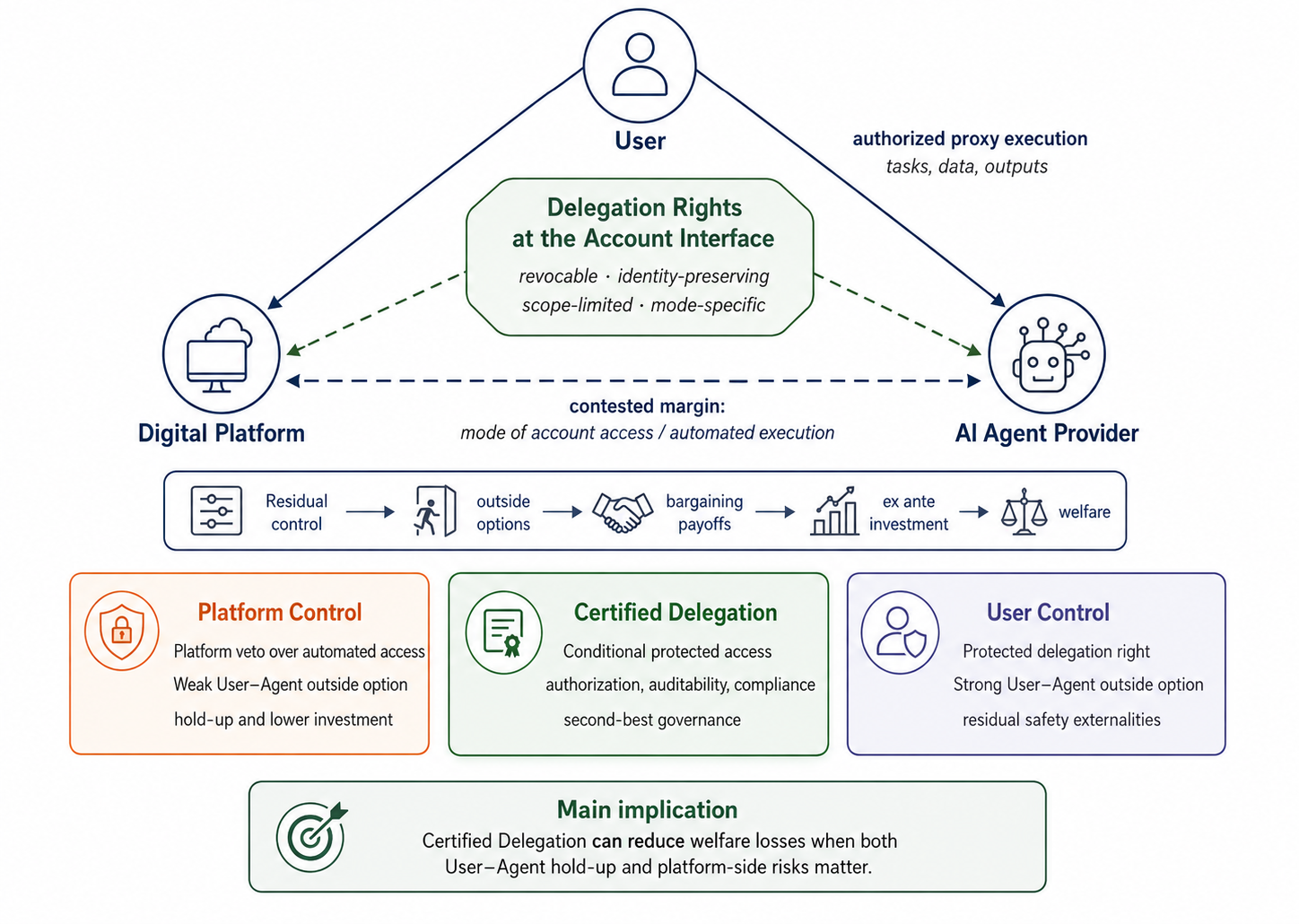}
\caption{Delegation rights at the account interface. The figure summarizes the three-party control problem among the User, Digital Platform, and AI Agent Provider. Delegation rights define a revocable, identity-preserving, scope-limited, and mode-specific authority for authorized proxy execution. The contested margin is the mode of account access and automated execution, which shapes residual control, outside options, bargaining payoffs, ex ante investment, and welfare. Platform Control preserves platform veto power but weakens the User--Agent outside option. User Control strengthens delegation but leaves residual safety externalities. Certified Delegation conditions protected access on authorization, auditability, and compliance. Certification can reduce welfare losses when both User--Agent hold-up and platform-side risks matter.}
\label{fig:theoretical_framework}
\end{figure}

\section{Institutional Foundation: AI Agents and the Delegation Gap}
\label{sec:institutional}

This section clarifies the institutional object studied in the paper. The problem is not whether platforms own their infrastructure, whether users hold ordinary account rights, or whether regulators should mandate data portability. The narrower question is whether an account holder may choose an automated agent as the means of exercising rights that she already holds. Existing digital-rights categories do not answer this question directly. We refer to this missing control margin as the \emph{delegation gap}.

The delegation gap arises because AI agents occupy an intermediate position. They are not ordinary users, because they act through software rather than direct human operation. They are not independent complementors, because their authority derives from the account holder rather than from a separate platform agreement. They are not merely data-portability tools, because they operate inside an ongoing account environment rather than only exporting information. This hybrid status makes the allocation of control over automated account operation economically important and legally unsettled.

\subsection{AI Agents as Account-Level Operators}
\label{subsec:account_level_operators}

Conventional digital tools usually remain outside the user's private account boundary. A search engine indexes public information and returns links. A browser extension may change how a page is displayed on the user's device, but it typically leaves the user as the party who decides and executes each transaction. A platform-approved API client interacts with the platform through a technical interface designed, limited, and monitored by the platform itself.

AI agents differ from these tools in a simple but important way: they can operate as account-level proxies. Once authorized by a user, an agent may read information, compare alternatives, fill forms, draft messages, make recommendations, and coordinate multi-step workflows within the user's account environment. The agent's role is therefore not limited to information retrieval or interface display. It may become the operational channel through which the user's account rights are exercised.

This does not mean that the agent acquires an independent right to enter the platform. Its authority is derivative. The agent acts for the account holder and within the scope of the account holder's permission. The institutional question is therefore not whether a third party may demand access to the platform on its own behalf. It is whether a user who can manually perform an action may delegate the execution of that same action to a machine proxy.

This shift changes the relevant margin of governance. The dispute is not primarily about the substantive content of the user's right: whether she may read a message, search a listing, place an order, or manage a booking. The dispute concerns the \emph{mode of execution}. Platforms may wish to restrict that mode because automated operation can create security, privacy, congestion, fraud-detection, or commercial-design risks. Users and agent providers may object because a categorical ban on automation can block valuable workflow innovations. The resulting conflict is precisely the delegation gap.

\subsection{Definition of Delegation Rights}
\label{subsec:definition_delegation_rights}

\begin{definition}[Delegation Right]
A \emph{delegation right} is the revocable, identity-preserving, scope-limited, and mode-specific right of an account holder to authorize an automated machine proxy to exercise pre-existing account-level usage entitlements on her behalf.
\end{definition}

The definition has four components.

First, the right is \emph{revocable}. The user must be able to withdraw authorization and terminate the agent's access. Without revocation, delegation begins to resemble account transfer or uncontrolled third-party access.

Second, the right is \emph{identity-preserving}. The agent does not replace the account holder as the contractual principal. It acts as a representative of the user, not as a new account owner or an independent platform participant.

Third, the right is \emph{scope-limited}. The agent may exercise only those actions that fall within the user's existing account privileges and the authorization granted by the user. Delegation does not expand the substantive bundle of platform rights.

Fourth, the right is \emph{mode-specific}. It governs the operational vehicle through which an existing right is exercised. It does not create a general right to access platform infrastructure, extract data, bypass security systems, or demand a new API.

These qualifications are essential. They make delegation rights narrower than a broad right to automate, scrape, or interoperate, but broader than a purely platform-discretionary permission. Delegation rights allocate residual control over one specific question: whether safe, authorized, and accountable machine-mediated execution is a permissible way for a user to use her account.

\subsection{Delegation and Adjacent Digital Rights}
\label{subsec:delegation_adjacent_rights}

Delegation rights are easiest to define by separating them from neighboring legal-economic categories.

They are not infrastructure ownership. A delegation right gives the user no control over the platform's servers, source code, databases, matching algorithms, authentication systems, or governance architecture. Those assets remain controlled by the platform. The delegation claim is marginal: conditional on the user already holding a valid account right, who controls the choice of execution interface?

Delegation rights are also not ordinary use rights. Ordinary use rights determine whether a user may enter the platform and consume its services. They usually say little about whether the user's actions must be performed manually or may be carried out by an authorized proxy. AI agents reveal this uncontracted margin.

Nor are delegation rights equivalent to account transferability. Transferability changes the identity of the entitlement holder. Delegation preserves that identity. The user remains the principal, while the agent acts within the user's account scope.

Delegation is also distinct from data portability. Portability governs the extraction and transmission of data, often from one service to another. Delegation governs operation within an ongoing platform environment. A platform can comply with data-portability obligations while still prohibiting automated account operation. Conversely, an agent can help a user operate her account without exporting the account's data to another platform.

Finally, delegation rights differ from platform-governed API access. API access is usually designed and licensed by the platform. The platform defines the endpoint, the rate limits, the data fields, the authentication method, and the developer terms. Delegation rights start from the user's authorization rather than from the platform's developer program. A certified delegation regime may use APIs as its technical implementation, but the prior question is institutional: should a user-authorized proxy be protected as a permissible mode of account use when it satisfies safety and accountability conditions?

Table~\ref{tab:delegation_gap} summarizes these distinctions.

\begin{table}[htbp]
\centering
\small
\caption{Delegation Rights and Adjacent Legal-Economic Categories}
\label{tab:delegation_gap}
\renewcommand{\arraystretch}{1.18}
\begin{tabular}{p{0.20\textwidth} p{0.29\textwidth} p{0.43\textwidth}}
\hline
\textbf{Category} & \textbf{Core Control Object} & \textbf{Why It Does Not Resolve AI-Agent Delegation} \tabularnewline
\hline
Ownership & Platform infrastructure and algorithmic capital & Delegation transfers no ownership over servers, source code, databases, matching systems, or core governance architecture. \tabularnewline
Ordinary use & Permission to enter and consume platform services & Ordinary use rights determine whether the user may use the service, but they do not specify whether execution must be human or may be delegated to a machine proxy. \tabularnewline
Account transferability & Assignment of the account or contract to a new principal & Delegation preserves the user's identity. The proxy acts on behalf of the existing account holder rather than becoming a new entitlement holder. \tabularnewline
Data portability & Extraction and movement of data across services & Portability concerns data transfer, not continuous operation inside the account interface. A platform may permit data export while still prohibiting automated account operation. \tabularnewline
API access & Platform-defined programmatic interface for external developers & API access is structured around platform permission and technical endpoints. Delegation begins from user authorization and asks whether that authorization receives protection. \tabularnewline
Delegation rights & Automated proxy execution of existing account privileges & Delegation rights isolate the residual-control question over the permissible mode of account operation. \tabularnewline
\hline
\end{tabular}
\end{table}

\subsection{Two Motivating Disputes}
\label{subsec:motivating_disputes}

Two stylized disputes help illustrate why the same delegation question can have different welfare implications across platform environments. The point is not to adjudicate the legal merits of any particular case. Rather, the examples show why a uniform rule is unlikely to be optimal.

\paragraph{AI-assisted commerce.}
In AI-assisted commerce, delegation can create substantial user-side surplus. A shopping agent may compare listings, normalize product attributes, check prices across sellers, summarize reviews, identify substitutes, and map choices to the user's preferences. These tasks are time-consuming for humans and often require repeated search across interfaces. In this environment, the agent's contribution is mainly productive: it reduces search costs and improves matching.

The platform, however, may still face legitimate concerns. Automated account operation can affect advertising placement, distort impression metrics, increase infrastructure load, weaken fraud-detection signals, or conflict with existing terms of service. The Amazon--Perplexity dispute is useful as an archetype because it places these two forces in tension: user-agent automation may generate meaningful consumer value, while the host platform may face operational and commercial risks.

\paragraph{Closed social ecosystems.}
Closed social ecosystems raise a different set of concerns. A machine proxy inside a messaging or social-network account may summarize conversations, draft replies, organize contacts, or coordinate group activity. These functions may benefit the delegating user, but they also touch information supplied by non-delegating parties. Messages, contact networks, group membership, social context, and interpersonal signals are not purely private inputs controlled by the delegating account holder.

The risks therefore extend beyond the user-agent relationship. Automated operation may expose sensitive communications, alter expectations of authenticity, or blur whether a message reflects human intent or machine-generated response. The WeChat--Doubao dispute serves as an archetype of this region of the problem: identity, privacy, and social-graph integrity carry greater weight than in a standard shopping environment.

The contrast between commerce platforms and closed social ecosystems illustrates the central institutional lesson. Delegation rights should not be governed by a simple binary rule. The welfare effects of protecting automated delegation depend on the value created by the user-agent coalition, the risks imposed on the platform and third parties, and the feasibility of verifying safe operation.

\subsection{From Institutional Conflict to Economic Model}
\label{subsec:institutional_to_model}

The delegation gap maps naturally into an incomplete-contracts framework. The relevant residual-control object is the permissible mode of account execution. This object is difficult to allocate fully by ex-ante contract because the relevant states change over time: platform interfaces are redesigned, agent capabilities improve, security vulnerabilities emerge, and regulatory expectations evolve.

The control allocation determines the parties' disagreement payoffs. Under Platform Control, the platform can block automated proxy execution if negotiations fail. The User--Agent coalition then loses the value of machine-mediated operation within the account, which weakens its bargaining position after relationship-specific investments have already been made. Under User Control, the user's delegation right is protected, so the User--Agent coalition retains a positive outside option through protected proxy execution, manual fallback, or technical workarounds.

These disagreement payoffs affect surplus division, and surplus division affects investment. If the platform holds an unconditional veto, users and agents may underinvest in workflow design, adaptation, and compliance because they expect part of the return to be appropriated ex post. If users hold an unconditional delegation right, the user-agent coalition may invest more, but the agent may not fully internalize the safety, privacy, or congestion costs imposed on the platform and third parties.

The model below formalizes this trade-off. We first compare the two polar regimes, Platform Control and User Control. We then introduce Certified Delegation as a second-best mechanism that protects user-authorized proxy operation only when the agent satisfies verifiable risk-mitigation and accountability requirements.

\section{Related Literature}
\label{sec:lit}

This paper is related to five bodies of work: property-rights theory, platform governance, data portability and interoperability, electronic-agent authorization, and AI-agent safety. Each literature speaks to part of the problem studied here. Property-rights theory explains why residual control matters when investments are non-contractible. Platform economics studies how digital platforms govern access to their ecosystems. Portability and interoperability rules create user-facing access rights, but mainly for data transfer or standardized interfaces. Cyberlaw and security engineering address authorization and machine execution. AI-safety work studies how automated agents can be evaluated and audited. The gap is that none of these literatures directly models the user's ability to authorize an automated proxy to operate inside an existing account. This paper treats that ability as a distinct residual-control margin.

\subsection{Property Rights, Incomplete Contracts, and Delegated Control}
\label{subsec:lit_prop_rights}

The model builds on the property-rights approach to incomplete contracts. \citet{KleinCrawfordAlchian1978} and \citet{Williamson1985} emphasize hold-up, appropriable quasi-rents, and the limits of complete contracting. \citet{GrossmanHart1986}, \citet{HartMoore1990}, and \citet{Hart1995} show that the allocation of residual control affects ex-post bargaining and, through that channel, ex-ante relationship-specific investment.

The same logic applies naturally to account-level AI delegation. The relevant control object, however, is not a factory, a physical asset, or a firm boundary. It is the permissible mode of using an account interface. A user may have the formal right to use an account, but the platform may retain practical authority over whether that use must be human-operated or may be carried out by an automated proxy. In this sense, the delegation problem is close to the distinction between formal and real authority in \citet{AghionTirole1997}. The user remains the contractual account holder, but the platform's control over execution mode may determine the value of the user's delegated use right.

The paper also uses a cooperative bargaining structure based on \citet{Shapley1953}, with related motivation from work on bargaining, multi-party surplus division, and hold-up (\citealp{HolmstromTirole1989}; \citealp{Tirole1999}). This allows the model to track how Platform Control, User Control, and Certified Delegation change coalition values and hence investment incentives. The contribution is not a new bargaining solution. It is to apply residual-control logic to a new object: account-level proxy execution.

\subsection{Platform Governance, Gatekeeping, and Complementor Access}
\label{subsec:lit_platforms}

A second related literature studies platforms as multi-sided markets and private governors of digital ecosystems. \citet{RochetTirole2003}, \citet{RochetTirole2006}, \citet{Armstrong2006}, and \citet{ParkerVanAlstyne2005} provide the standard framework for two-sided pricing and cross-side network effects. Work on platform openness and complementor access, including \citet{Boudreau2010}, \citet{BoudreauHagiu2009}, \citet{EisenmannParkerVanAlstyne2009}, and \citet{TiwanaKonsynskiBush2010}, shows how platforms use technical and contractual rules to shape participation by developers, merchants, and other complementors. Policy work such as \citet{CremerMontjoyeSchweitzer2019} and the \citet{StiglerReport2019} further highlights the gatekeeping power of large digital platforms.

That literature usually studies access from the standpoint of complementors or competing services. The delegation problem is different. An AI agent operating inside an account does not necessarily enter the platform as an independent app, merchant, or data recipient. It may act as a representative of the user. The economic question is therefore not only whether platforms should open interfaces to outsiders, but whether a user can choose the means by which her own account rights are exercised.

This distinction matters for both competition and governance. A platform may have legitimate reasons to restrict automated operation, including fraud prevention, privacy protection, and system-load management. At the same time, a categorical ban on user-authorized agents may protect the platform from a new layer of user-side intermediation. The model captures this tension by treating proxy delegation as a residual-control problem rather than as a standard complementor-access problem.

\subsection{Data Portability, Interoperability, and the Limits of Access Mandates}
\label{subsec:lit_portability}

The paper also relates to work on data portability and interoperability. In law and policy, Article 20 of the GDPR \citep{GDPR2016}, the Digital Markets Act \citep{DMA2022}, the Digital Services Act \citep{DSA2022}, and the Data Act \citep{DataAct2023} create or strengthen rights to data access, transfer, or platform interoperability. Economic and legal analyses by \citet{SwireLagos2013}, \citet{GraefVerschakelenValcke2013}, \citet{KramerStuedlein2019}, \citet{Kramer2021}, \citet{KramerSenellartDeStreel2020}, and \citet{OECD2021Portability} examine how such mandates affect switching costs, contestability, privacy, and platform incentives.

Delegation rights are related to these access rights but not identical to them. Portability concerns the movement of data out of one environment and into another. Interoperability typically concerns standardized connections between systems or developer-facing interfaces. By contrast, account-level delegation concerns continuous operation inside an existing account environment. A platform may comply with portability obligations while still banning automated account operation. Conversely, an agent may help a user act inside an account without exporting the user's data to a competing service.

The paper therefore treats delegation as a complement to portability and interoperability, not as a substitute. The relevant question is not only whether data can move, or whether an interface must be exposed. It is whether a user-authorized proxy can exercise existing account privileges when it satisfies appropriate limits on authorization, scope, and risk.

\subsection{Electronic Agents, Authorization Protocols, and Access Legality}
\label{subsec:lit_authorization}

A separate line of work studies electronic agents, delegated authorization, and access legality. In security engineering, OAuth and related standards (\citealp{RFC6749}; \citealp{RFC9396}; \citealp{RFC9700}) provide technical tools for delegated authorization, scope-limited tokens, and revocation. These protocols are important because they show that delegation can be made granular and revocable. They do not, however, answer the legal or economic question of whether user-authorized machine operation should receive access protection over a platform's objection.

Legal scholarship has examined related questions. \citet{Weitzenboeck2004} discusses when actions by electronic agents bind a principal. \citet{Calo2015} analyzes how automated systems strain traditional cyberlaw categories. Recent U.S. cases also show that the boundary between authorized access, terms-of-service restrictions, and automated interaction remains contested. \citet{VanBuren2021} narrowed the interpretation of ``exceeds authorized access'' under the CFAA, while \citet{hiQLinkedIn2022} addressed automated access to publicly available platform data.

These debates do not map perfectly onto the delegation problem. Much of the case law concerns scraping, public data, or access after objection by the host. The setting here is narrower: the user has an account and authorizes a proxy to act on her behalf within the scope of that account. The paper contributes by giving this setting an economic structure. It asks how different allocations of control over delegated execution affect bargaining and investment.

\subsection{AI-Agent Safety, Auditing Protocols, and Conditional Governance}
\label{subsec:lit_safety}

The paper is also connected to the emerging literature on AI-agent capability, evaluation, and auditing. Benchmarks such as \citet{WebArena2023} study whether language-model agents can complete multi-step web tasks. Commercial systems such as \citet{OpenAIOperator2025}, \citet{OpenAIChatGPTAgent2025}, and \citet{AnthropicComputerUse2024} illustrate the practical movement from chatbot interaction toward agents that operate software interfaces on behalf of users. Policy and standards documents such as \citet{NISTAIRMF2023}, \citet{NISTGenAIProfile2024}, \citet{ISO42001}, and the EU AI Act \citep{EUAIAct2024} emphasize risk management, monitoring, documentation, and accountability. Work on algorithmic auditing, including \citet{RajiEtAl2020}, provides a framework for evaluating systems after deployment.

This literature makes certification and auditing a plausible governance tool. The present paper asks what certification does in an economic model of platform control. In the model, certification is not only a safety screen. It also changes bargaining power. A certified proxy receives a protected access path; an uncertified proxy remains excludable. This turns technical verification into a conditional allocation of residual control.

The analysis therefore connects AI-agent safety to platform economics. Safety standards matter not only because they reduce expected harm directly, but also because they can provide a rule-based condition under which platforms lose the ability to exclude user-authorized agents arbitrarily.

\subsection{Positioning}
\label{subsec:lit_positioning}

The closest way to position this paper is as follows. Property-rights theory provides the mechanism: control rights affect bargaining and investment. Platform economics provides the setting: private platforms govern access to digital interfaces. Portability, interoperability, and cyberlaw provide related legal categories, but they focus mainly on data movement, public access, or platform-designed interfaces. AI-safety work provides the idea of certification and auditing, but usually does not model how certification changes investment incentives among users, agents, and platforms.

This paper combines these elements around a specific control object: the user's ability to delegate account operation to an automated proxy. The central claim is limited. Delegation rights should not be confused with a general right to scrape, a right to unrestricted API access, or a transfer of platform ownership. They concern the mode of exercising existing account rights. Certified Delegation is proposed as a second-best response to that problem: it protects certified, user-authorized proxies from arbitrary exclusion while preserving the platform's right to refuse unsafe or non-compliant automation.

\section{Theoretical Framework}
\label{sec:model}

This section presents a simple incomplete-contracts model of account-level delegation. The model has three parties: a User, an AI Agent provider, and a Platform. The purpose is not to estimate a particular platform dispute, but to isolate how the allocation of control over automated account operation affects disagreement payoffs, bargaining positions, relationship-specific investment, and welfare.

The central control object is the mode of account execution. The user already holds an account entitlement. The question is whether that entitlement must be exercised manually or may be exercised through a user-authorized automated proxy. Platform Control, User Control, and Certified Delegation differ in how they allocate residual control over this execution mode.

\subsection{Players, Investments, and Timing}
\label{subsec:model_players}

There are three risk-neutral parties,
\[
\mathcal{N}=\{U,A,P\},
\]
where $U$ denotes the User, $A$ denotes the AI Agent provider, and $P$ denotes the Platform. The user holds an account on the platform. The agent can act as an automated proxy inside that account when delegation is permitted. The platform owns and maintains the infrastructure on which the account operates.

Before bargaining, each party chooses a non-contractible, relationship-specific investment:
\[
i_U\geq 0,\qquad i_A\geq 0,\qquad i_P\geq 0.
\]
The user's investment $i_U$ captures workflow design, preference specification, and monitoring of delegated tasks. The agent's investment $i_A$ captures platform-specific execution capability, including reliability and compliance-related adaptation. The platform's investment $i_P$ captures infrastructure, authentication, fraud prevention, and system stability.

Investment costs are quadratic:
\begin{equation}
C_U(i_U)=\frac{\kappa_U}{2}i_U^2,\qquad
C_A(i_A)=\frac{\kappa_A}{2}i_A^2,\qquad
C_P(i_P)=\frac{\kappa_P}{2}i_P^2,
\label{eq:costs}
\end{equation}
where $\kappa_U,\kappa_A,\kappa_P>0$.

The timing is as follows:
\begin{enumerate}
    \item \textbf{Governance choice.} A regime $\Omega$ assigns residual control over account-level proxy delegation.
    \item \textbf{Investment.} The three parties choose $(i_U,i_A,i_P)$. These investments are sunk before bargaining.
    \item \textbf{Bargaining.} The parties bargain over the surplus from automated account use. Payoffs are allocated according to the Shapley value.
    \item \textbf{Realization.} Production takes place and each party receives its net payoff.
\end{enumerate}

\subsection{Gross Surplus and Security Externality}
\label{subsec:model_surplus}

When all three parties cooperate, automated account use generates gross surplus
\begin{equation}
V(i_U,i_A,i_P)
=
B i_U^{\alpha_U}i_A^{\alpha_A}i_P^{\alpha_P},
\label{eq:gross_value}
\end{equation}
where $B>0$ and
\[
\alpha_U>0,\qquad
\alpha_A>0,\qquad
\alpha_P>0,\qquad
\alpha_U+\alpha_A+\alpha_P<1.
\]
The decreasing-returns assumption ensures a well-behaved interior optimum under the convex cost functions. The multiplicative specification captures complementarity among user adaptation, agent capability, and platform infrastructure.

Automated proxy use may also impose costs on the platform and on third parties. These costs may arise from privacy leakage, fraud risk, system load, or security vulnerabilities. We represent this externality by
\begin{equation}
K(i_A)
=
\frac{\rho_0}{1+\rho_1 i_A},
\label{eq:security_cost}
\end{equation}
where $\rho_0>0$ is baseline risk and $\rho_1>0$ measures the effectiveness of agent-side risk mitigation. Hence $K'(i_A)<0$ and $K''(i_A)>0$.

The net value of the grand coalition is
\begin{equation}
V_G(i_U,i_A,i_P)
=
V(i_U,i_A,i_P)-K(i_A).
\label{eq:grand_value}
\end{equation}

\subsection{Reduced-Form Interpretation of Agent Investment}
\label{subsec:reduced_form_agent_investment}

The agent's investment $i_A$ should be read as a reduced-form measure of platform-specific capability. It combines two margins. The first is productive capability, such as interface adaptation, workflow execution, and task reliability. The second is compliance capability, such as authorization controls, audit logs, rate-limit adherence, and data minimization.

This scalar treatment keeps the bargaining mechanism transparent. It also imposes an interpretive limitation. Because $i_A$ combines productive adaptation and compliance capability, the calibration cannot separately identify the production-recovery channel from the safety-compliance channel at the level of primitive investments. The numerical results should therefore be interpreted as illustrating the joint mechanism rather than decomposing distinct engineering margins. A richer model could separate productive capability from safety capability, but the present formulation is designed to isolate the residual-control mechanism with minimal notation.

\subsection{Control Regimes}
\label{subsec:model_control_regimes}

We compare two benchmark regimes:
\[
\Omega_P:\text{ Platform Control},
\qquad
\Omega_U:\text{ User Control}.
\]

Under \emph{Platform Control} ($\Omega_P$), the platform has residual authority over whether automated proxy execution is permitted. If bargaining fails, the platform can block machine-mediated account use. The User--Agent coalition then cannot realize the value of automated workflows inside the platform account.

Under \emph{User Control} ($\Omega_U$), the user has a protected delegation right. If bargaining fails, the User--Agent coalition can still operate through the automated proxy, either through direct account privileges, manual fallback, or other non-cooperative access channels.

The regime does not change the production technology $V(\cdot)$ or the externality function $K(\cdot)$. It changes the coalition values that determine bargaining payoffs.

\subsection{Coalition Values and Bargaining}
\label{subsec:model_bargaining}

The bargaining game assigns a value to each coalition of parties. These coalition values are reduced-form objects. They represent what each coalition can achieve if it forms without the excluded party. Table~\ref{tab:coalition_values_interpretation} summarizes the economic interpretation of the relevant coalitions.

\begin{table}[htbp]
\centering
\small
\caption{Economic Interpretation of Coalition Values}
\label{tab:coalition_values_interpretation}
\renewcommand{\arraystretch}{1.18}
\begin{tabularx}{\textwidth}{p{0.25\textwidth} X X}
\toprule
\textbf{Coalition} & \textbf{Economic Interpretation} & \textbf{Example} \\
\midrule
$U+P$ 
& Manual or platform-native account use without a third-party AI agent. The user and platform can generate value, but the agent's platform-specific automation capability is absent. 
& The user manually shops, books, searches, pays, or sends messages through the ordinary platform interface. \\
\addlinespace
$A+P$ 
& Platform-approved automation, embedded assistant services, developer integration, or back-end automation without the user's relationship-specific workflow investment. 
& A platform-integrated agent service, official API partnership, or platform-side automation tool. \\
\addlinespace
$U+A$ under $\Omega_P$ 
& No protected access to the account interface. The platform can block automated proxy execution, so the coalition value is normalized to zero. 
& A user-authorized agent is excluded because the platform refuses automated access. \\
\addlinespace
$U+A$ under $\Omega_U$ 
& Fallback or protected proxy execution without platform cooperation. The coalition can realize a reduced value through user-side access channels. 
& Browser automation, manual fallback supported by the agent, or another user-side execution channel. \\
\addlinespace
$U+A+P$ 
& Full cooperative delegated account use. The user, agent provider, and platform jointly realize the value of automated account operation while bearing the associated risk. 
& Certified delegation, negotiated integration, or another cooperative access arrangement. \\
\bottomrule
\end{tabularx}
\end{table}

The $A$--$P$ coalition deserves special clarification. It should not be read as account-level delegation without a user. Delegation, as defined in this paper, remains derivative of user authorization. Instead, $v_{AP}$ captures platform-approved automation or integration services that use agent capability without the user's relationship-specific workflow investment. Examples include an official API partnership, an embedded platform assistant, or a back-end automation service supplied by the agent provider to the platform. This interpretation keeps the $A$--$P$ coalition consistent with the paper's definition of delegation rights.

Singleton coalition values are normalized to zero. Under any regime $\Omega$, the value of the grand coalition is
\begin{equation}
v(\{U,A,P\}\mid \Omega)=V_G(i_U,i_A,i_P).
\label{eq:grand_coalition_value}
\end{equation}

Two-party coalitions involving the platform generate reduced surplus:
\begin{align}
v_{UP}(i_U,i_P)
&=
\beta_{UP}B i_U^{\alpha_U}i_P^{\alpha_P},
\qquad 0<\beta_{UP}<1,
\label{eq:up_value} \\
v_{AP}(i_A,i_P)
&=
\beta_{AP}B i_A^{\alpha_A}i_P^{\alpha_P},
\qquad 0<\beta_{AP}<1.
\label{eq:ap_value}
\end{align}
The parameters $\beta_{UP}$ and $\beta_{AP}$ capture the loss of trilateral complementarity when one party is absent.

The key regime-dependent term is the value of the User--Agent coalition:
\begin{equation}
v_{UA}^{\Omega}(i_U,i_A)
=
\begin{cases}
0, & \Omega=\Omega_P, \\[0.35em]
\lambda B i_U^{\alpha_U}i_A^{\alpha_A}, & \Omega=\Omega_U,
\end{cases}
\label{eq:ua_value}
\end{equation}
where $\lambda\in(0,1)$ measures the effectiveness of protected delegation, manual fallback, browser-side automation, or other non-cooperative user-side access channels when platform cooperation is absent.

Payoffs are allocated by the Shapley value. For each party $k\in\{U,A,P\}$,
\begin{equation}
\varphi_k^\Omega(i)
=
\sum_{S\subseteq \mathcal{N}\setminus\{k\}}
\frac{|S|!(2-|S|)!}{3!}
\left[
v(S\cup\{k\}\mid\Omega)-v(S\mid\Omega)
\right].
\label{eq:shapley_general}
\end{equation}
Substituting the coalition values gives
\begin{align}
\varphi_U^\Omega
&=
\frac{1}{6}
\left[
2V_G+v_{UA}^{\Omega}+v_{UP}-2v_{AP}
\right],
\label{eq:phi_u} \\
\varphi_A^\Omega
&=
\frac{1}{6}
\left[
2V_G+v_{UA}^{\Omega}+v_{AP}-2v_{UP}
\right],
\label{eq:phi_a} \\
\varphi_P^\Omega
&=
\frac{1}{6}
\left[
2V_G+v_{UP}+v_{AP}-2v_{UA}^{\Omega}
\right].
\label{eq:phi_p}
\end{align}

Each party chooses its investment to maximize its bargaining payoff net of investment cost:
\begin{equation}
R_k^\Omega(i)=\varphi_k^\Omega(i)-C_k(i_k),
\qquad k\in\{U,A,P\}.
\label{eq:private_payoff}
\end{equation}
An interior Nash equilibrium under regime $\Omega$ satisfies
\begin{equation}
\frac{\partial \varphi_k^\Omega(i)}{\partial i_k}
=
C_k'(i_k),
\qquad k\in\{U,A,P\}.
\label{eq:private_foc}
\end{equation}

As a benchmark, the first-best allocation solves
\begin{equation}
\max_{i_U,i_A,i_P\geq 0} W(i_U,i_A,i_P),
\label{eq:first_best_problem}
\end{equation}
where
\begin{equation}
W(i_U,i_A,i_P)
=
V(i_U,i_A,i_P)-K(i_A)
-C_U(i_U)-C_A(i_A)-C_P(i_P).
\label{eq:social_welfare}
\end{equation}
The first-order conditions are
\begin{align}
\frac{\partial V}{\partial i_U} &= C_U'(i_U),
\label{eq:fb_u} \\
\frac{\partial V}{\partial i_A} - K'(i_A) &= C_A'(i_A),
\label{eq:fb_a} \\
\frac{\partial V}{\partial i_P} &= C_P'(i_P).
\label{eq:fb_p}
\end{align}

\subsection{Regularity Assumptions}
\label{subsec:model_assumptions}

The comparative-static results below require standard regularity conditions. These assumptions are maintained throughout the theoretical analysis and are imposed only on the relevant parameter region considered by the model.

\begin{assumption}[Interiority]
\label{ass:interiority}
For each regime $\Omega\in\{\Omega_P,\Omega_U\}$, the investment game has an interior equilibrium
\[
i^\Omega=(i_U^\Omega,i_A^\Omega,i_P^\Omega)
\]
in the parameter region considered.
\end{assumption}

\begin{assumption}[Stability]
\label{ass:stability}
The equilibrium of the investment game is locally stable. Equivalently, small changes in marginal payoff functions induce small changes in the equilibrium investment vector rather than discontinuous jumps across equilibria.
\end{assumption}

\begin{assumption}[Strategic Complementarity]
\label{ass:strategic_complementarity}
The relevant surplus components exhibit non-negative cross-partial derivatives in the parameter region considered. In particular, increases in one party's relationship-specific investment weakly raise the marginal return to complementary investments by the other parties.
\end{assumption}

\begin{assumption}[Regularity]
\label{ass:regularity}
The equilibrium investment correspondence is single-valued and continuous in the parameter region considered. The welfare functions induced by equilibrium investments are also continuous in the relevant primitive parameters.
\end{assumption}

Assumption~\ref{ass:interiority} allows the investment comparison to be expressed through first-order conditions. Assumption~\ref{ass:stability} rules out unstable comparative statics. Assumption~\ref{ass:strategic_complementarity} is consistent with the multiplicative surplus specification and captures the idea that user adaptation, agent capability, and platform infrastructure are complements. Assumption~\ref{ass:regularity} ensures that local shifts in control rights can be mapped into local shifts in equilibrium investment and welfare.

\subsection{Platform Control and AI-Complementary Hold-Up}
\label{subsec:platform_holdup}

Under Platform Control, $v_{UA}^{\Omega_P}=0$. The User--Agent coalition therefore has no protected disagreement payoff if bargaining fails. Under User Control, by contrast,
\[
v_{UA}^{\Omega_U}(i_U,i_A)=\lambda B i_U^{\alpha_U}i_A^{\alpha_A}>0
\]
for any interior investment profile.

The first step is to compare marginal incentives across the two regimes.

\begin{lemma}[Delegation Access and User--Agent Marginal Incentives]
\label{lem:marginal_incentives}
For any interior investment profile, User Control raises the private marginal return to user and agent investment relative to Platform Control:
\begin{align}
\frac{\partial \varphi_U^{\Omega_U}}{\partial i_U}
-
\frac{\partial \varphi_U^{\Omega_P}}{\partial i_U}
&=
\frac{1}{6}
\frac{\partial v_{UA}^{\Omega_U}}{\partial i_U}
>0,
\label{eq:marginal_u_difference} \\
\frac{\partial \varphi_A^{\Omega_U}}{\partial i_A}
-
\frac{\partial \varphi_A^{\Omega_P}}{\partial i_A}
&=
\frac{1}{6}
\frac{\partial v_{UA}^{\Omega_U}}{\partial i_A}
>0.
\label{eq:marginal_a_difference}
\end{align}
\end{lemma}

\begin{proof}
The production technology, externality function, and platform-involving coalition values are the same under $\Omega_P$ and $\Omega_U$. The only term in equations~\eqref{eq:phi_u} and \eqref{eq:phi_a} that differs across the two regimes is $v_{UA}^{\Omega}$. Under Platform Control, this term is zero. Under User Control, it is positive and strictly increasing in $i_U$ and $i_A$ at any interior investment profile. Taking derivatives gives equations~\eqref{eq:marginal_u_difference} and \eqref{eq:marginal_a_difference}.
\end{proof}

Lemma~\ref{lem:marginal_incentives} establishes the local incentive effect. The next proposition translates this marginal comparison into an equilibrium comparison under the regularity assumptions.

\begin{proposition}[Platform Control and AI-Complementary Hold-Up]
\label{prop:platform_holdup}
Under Assumptions~\ref{ass:interiority}--\ref{ass:regularity}, Platform Control lowers the private marginal return to user and agent investment relative to User Control. In a stable interior equilibrium, this implies weakly lower equilibrium investment by the User--Agent coalition:
\[
i_U^{\Omega_P}\leq i_U^{\Omega_U},
\qquad
i_A^{\Omega_P}\leq i_A^{\Omega_U},
\]
with strict inequalities when the delegation-disagreement payoff is locally payoff relevant.
\end{proposition}

\begin{proof}
By Lemma~\ref{lem:marginal_incentives}, removing the protected User--Agent disagreement payoff shifts down the marginal payoff functions for $i_U$ and $i_A$. Assumption~\ref{ass:strategic_complementarity} ensures that these lower marginal returns are not reversed through negative cross-investment effects. Assumptions~\ref{ass:interiority}, \ref{ass:stability}, and \ref{ass:regularity} allow the comparison of stable interior equilibria through the first-order conditions in equation~\eqref{eq:private_foc}. Therefore, the equilibrium investments of the user and the agent are weakly lower under Platform Control than under User Control. The inequalities are strict when $v_{UA}^{\Omega_U}$ has a strictly positive local marginal effect on the relevant payoff.
\end{proof}

The relation to the first-best benchmark should be interpreted with care. The formal comparison above is between Platform Control and User Control. Relative to the first-best allocation, both regimes generally distort relationship-specific investment because each party captures only a fraction of the marginal surplus while bearing the full private cost of investment. Platform Control exacerbates this distortion for the User--Agent coalition by eliminating the delegation-based disagreement payoff.

\subsection{User Control and Safety Externalities}
\label{subsec:user_control_externalities}

User Control strengthens the bargaining position of the User--Agent coalition, but it does not necessarily internalize all risks created by automated account operation. The social marginal return to agent investment includes both its productive effect and its effect on expected risk:
\[
\frac{\partial V}{\partial i_A}-K'(i_A).
\]
The agent's private payoff, however, places only fractional weight on the risk-reduction component. From equation~\eqref{eq:phi_a}, the term $-K'(i_A)$ enters the agent's payoff through the grand-coalition component rather than as a fully internalized benefit.

\begin{proposition}[User Control and Safety Externalities]
\label{prop:user_control_externalities}
User Control raises the private marginal returns to user and agent investment relative to Platform Control and therefore mitigates the hold-up problem. However, absent perfectly offsetting transfers, liability rules, or certification requirements, User Control generally does not implement the first-best incentive to reduce safety, privacy, congestion, or third-party risks.
\end{proposition}

\begin{proof}
The increase in productive investment incentives follows directly from Lemma~\ref{lem:marginal_incentives}. The externality distortion follows from comparing the agent's private first-order condition in equation~\eqref{eq:private_foc} with the planner's first-order condition in equation~\eqref{eq:fb_a}. The planner places full weight on the marginal risk-reduction benefit $-K'(i_A)$. The agent receives only the bargaining-weighted component of that benefit through the grand-coalition value. Unless transfers, liability, or certification rules make the agent internalize the remaining risk, User Control does not generally implement the first-best safety incentive.
\end{proof}

Thus, User Control addresses the platform hold-up problem but leaves a residual externality problem. This motivates a conditional mechanism in which protected delegation depends on verifiable risk control.

\subsection{Welfare Decomposition}
\label{subsec:welfare_decomposition}

Let
\[
i^{\Omega_P} = (i_U^{\Omega_P},i_A^{\Omega_P},i_P^{\Omega_P})
\]
and
\[
i^{\Omega_U} = (i_U^{\Omega_U},i_A^{\Omega_U},i_P^{\Omega_U})
\]
denote the equilibrium investment vectors under Platform Control and User Control. Welfare under regime $\Omega$ is
\begin{equation}
W(\Omega)
=
V(i_U^\Omega,i_A^\Omega,i_P^\Omega)
-
K(i_A^\Omega)
-
C_U(i_U^\Omega)
-
C_A(i_A^\Omega)
-
C_P(i_P^\Omega).
\label{eq:regime_welfare}
\end{equation}

Define the welfare difference
\[
\Delta W = W(\Omega_U)-W(\Omega_P).
\]
To separate the main channels, introduce the counterfactual vector
\[
\widetilde{i} = (i_U^{\Omega_U},i_A^{\Omega_U},i_P^{\Omega_P}).
\]
Then
\begin{equation}
\Delta W
=
\Delta W^{UA}
+
\Delta W^{P}
+
\Delta W^{K}.
\label{eq:welfare_decomposition}
\end{equation}
The three components are
\begin{align}
\Delta W^{UA}
&=
\left[ V(\widetilde{i})-V(i^{\Omega_P}) \right]
-
\left[
C_U(i_U^{\Omega_U})+C_A(i_A^{\Omega_U})
- C_U(i_U^{\Omega_P})-C_A(i_A^{\Omega_P})
\right],
\label{eq:delta_w_ua} \\
\Delta W^{P}
&=
\left[ V(i^{\Omega_U})-V(\widetilde{i}) \right]
-
\left[
C_P(i_P^{\Omega_U})-C_P(i_P^{\Omega_P})
\right],
\label{eq:delta_w_p} \\
\Delta W^{K}
&=
-
\left[
K(i_A^{\Omega_U})-K(i_A^{\Omega_P})
\right].
\label{eq:delta_w_k}
\end{align}

The term $\Delta W^{UA}$ captures the welfare effect of changing user and agent investment. The term $\Delta W^{P}$ captures the effect of any change in platform investment. The term $\Delta W^{K}$ captures the change in expected risk. This decomposition is useful because the sign of $\Delta W$ is not determined by a single force. User Control may improve investment incentives but worsen risk internalization; Platform Control may reduce risk exposure but depress relationship-specific investment.

\subsection{Contribution-Threshold Logic as a Regime-Map Heuristic}
\label{subsec:contribution_threshold}

Define the User--Agent contribution index as
\begin{equation}
\Gamma=\alpha_U+\alpha_A.
\label{eq:gamma}
\end{equation}
A larger $\Gamma$ means that more of the surplus depends on user adaptation and agent capability. A larger $\alpha_P$ means that platform infrastructure is more important. Baseline risk is governed by $\rho_0$, and the effectiveness of outside options is governed by $\lambda$.

The model does not imply a universal closed-form threshold between Platform Control and User Control. Instead, it suggests a local contribution-threshold logic. User Control is more attractive when user and agent investments are central to value creation and when non-cooperative workarounds are effective. Platform Control is more attractive when platform infrastructure, identity protection, privacy, or third-party risks dominate the welfare calculation.

The following result formalizes the local threshold intuition rather than providing a global characterization.

\begin{proposition}[Local Regime-Map Heuristic]
\label{prop:contribution_threshold}
Fix $\alpha_P$, $\rho_0$, $\rho_1$, $\lambda$, and the investment cost parameters. Suppose the welfare difference $\Delta W(\Gamma)$ is continuous and satisfies a local single-crossing property around a point of indifference. Then there exists a local threshold
\[
\Gamma^*(\alpha_P,\rho_0,\rho_1,\lambda)
\]
such that User Control is locally welfare-superior when $\Gamma>\Gamma^*$, while Platform Control is locally welfare-superior when $\Gamma<\Gamma^*$.
\end{proposition}

\begin{proof}
By Assumption~\ref{ass:regularity}, equilibrium investments and induced welfare are continuous in the relevant primitive parameters. Hence $\Delta W(\Gamma)$ is continuous in the parameter region considered. If $\Delta W(\Gamma)$ satisfies a local single-crossing property around a point of indifference, the existence of a local threshold follows from the Intermediate Value Theorem. The direction of the local comparison follows from the sign of $\Delta W(\Gamma)$ on each side of the crossing.
\end{proof}

This result should be read as a regime-map heuristic. It does not establish a global welfare ranking, and it does not claim that the boundary is invariant across functional forms or institutional environments. Its role is to discipline the interpretation of the numerical regime maps: environments with high user-agent contribution and effective outside options tend to favor stronger delegation protection, while environments with high platform centrality and high third-party risk tend to favor stronger platform control. The conditional mechanism introduced below is motivated by precisely this lack of a universal polar-regime ranking.

\section{Certified Delegation Mechanism}
\label{sec:certified}

The previous section shows that the two benchmark regimes generate different distortions. Platform Control can weaken user and agent investment by giving the platform a strong ex-post veto over automated account operation. User Control reduces this hold-up problem, but it may leave some security, privacy, congestion, and third-party costs outside the agent's private objective. This section introduces a conditional regime, \emph{Certified Delegation}, designed to address both margins.

Certification has two effects. First, it screens or disciplines unsafe automation by imposing a verifiable risk standard. Second, and more importantly for the property-rights analysis, it changes the platform's residual right of refusal. Once the agent satisfies the standard, the platform can no longer exclude the proxy merely because execution is automated. If the agent fails the standard, the platform retains the right to refuse access. Certified Delegation is therefore not ordinary certification in the narrow technical sense. It is \emph{access-protecting certification}: compliance changes the allocation of residual control over the account interface.

Under Certified Delegation, a user-authorized agent receives protection against exclusion only if it satisfies verifiable requirements for authorization, revocability, auditability, risk control, and accountability. A platform may still refuse access to uncertified or non-compliant agents. The point of the mechanism is not to create an unconditional right to automate platform use. It is to distinguish between unsafe automation, which remains excludable, and certified delegation, which receives a protected access path.

As in Section~\ref{subsec:reduced_form_agent_investment}, the agent's investment $i_A$ is treated as a reduced-form measure of platform-specific capability. It includes both productive capability, such as interface adaptation and task execution, and compliance capability, such as authorization controls, audit logs, rate-limit adherence, and data minimization. This scalar specification keeps the bargaining mechanism transparent. It also limits the interpretation of the numerical exercises. Because $i_A$ combines productive adaptation and compliance capability, the calibration cannot separately identify the production-recovery channel from the safety-compliance channel at the level of primitive investments. The numerical results should therefore be interpreted as illustrating the joint mechanism rather than decomposing distinct engineering margins. A richer two-investment model would separate productive capability from safety capability; the present scalar specification is used to keep the residual-control mechanism tractable.

\subsection{Mechanism Design and the Compliance Threshold}
\label{subsec:cert-mechanism-threshold}

Let $\bar{K}>0$ denote a publicly announced risk threshold. A lower value of $\bar{K}$ corresponds to a stricter certification standard. The Certified Delegation regime is denoted by $\Omega_{\mathrm{Cert}}(\bar{K})$.

The rule is simple. If the agent satisfies the risk standard, $K(i_A)\leq \bar{K}$, the effective regime is User Control. If the agent fails the standard, $K(i_A)>\bar{K}$, the effective regime is Platform Control. Formally,
\begin{equation}
\Omega_{\mathrm{eff}}(i_A;\bar{K})=\Omega_U
\quad \text{if} \quad
K(i_A)\leq \bar{K},
\label{eq:cert-effective-user}
\end{equation}
and
\begin{equation}
\Omega_{\mathrm{eff}}(i_A;\bar{K})=\Omega_P
\quad \text{if} \quad
K(i_A)>\bar{K}.
\label{eq:cert-effective-platform}
\end{equation}

This switching rule is the core institutional feature of the mechanism. When the standard is met, the platform's discretion to refuse automation is limited. When the standard is not met, the platform retains residual control over exclusion. Certification therefore operates as a conditional property-rights rule rather than merely as a technical label.

Using the risk function
\[
K(i_A)=\frac{\rho_0}{1+\rho_1 i_A},
\]
the condition $K(i_A)\leq \bar{K}$ can be written as a minimum investment requirement:
\begin{equation}
i_A^c(\bar{K})
=
\max
\left\{
0,
\frac{\rho_0/\bar{K}-1}{\rho_1}
\right\},
\qquad
0<\bar{K}\leq \rho_0.
\label{eq:cert-threshold}
\end{equation}
Thus, the agent qualifies for access protection if and only if
\[
i_A\geq i_A^c(\bar{K}).
\]

The scalar threshold $\bar{K}$ summarizes a bundle of observable compliance requirements. These may include explicit user authorization, revocation, scope limitation, audit logs, rate-limit compliance, data minimization, incident reporting, and liability rules. Once these conditions are met, the platform cannot refuse access solely because the operator is automated. If the conditions are not met, the platform may refuse the proxy in order to protect its infrastructure, users, counterparties, and third parties.

\subsection{Agent Incentive Compatibility}
\label{subsec:cert-agent-ic}

Certification changes the agent's investment problem because access protection begins only after the threshold is reached. Let $F_c\geq 0$ denote a fixed certification cost. Let $c_c(i_A)$ denote variable compliance costs, with
\[
c_c'(i_A)\geq 0,
\qquad
c_c''(i_A)\geq 0.
\]
For given user and platform investments $(i_U,i_P)$, an uncertified agent remains under Platform Control. Its best payoff is
\begin{equation}
R_A^{P*}
=
\max_{i_A<i_A^c(\bar{K})}
\left[
\varphi_A^{\Omega_P}(i_U,i_A,i_P)-C_A(i_A)
\right].
\label{eq:cert-best-uncertified}
\end{equation}

Certification is privately attractive if the payoff from meeting the threshold is at least as large as the best payoff from remaining uncertified:
\begin{equation}
\varphi_A^{\Omega_U}(i_U,i_A^c,i_P)
-
C_A(i_A^c)
-
F_c
-
c_c(i_A^c)
\geq
R_A^{P*}.
\label{eq:cert-agent-ic}
\end{equation}

It is useful to define the delegation premium at the threshold:
\begin{equation}
\Delta_A^c(\bar{K})
=
\varphi_A^{\Omega_U}(i_U,i_A^c,i_P)
-
\varphi_A^{\Omega_P}(i_U,i_A^c,i_P).
\label{eq:cert-premium-definition}
\end{equation}
Using the Shapley payoffs derived above, this premium is
\begin{equation}
\Delta_A^c(\bar{K})
=
\frac{1}{6}
v_{UA}^{\Omega_U}(i_U,i_A^c).
\label{eq:cert-premium-shapley}
\end{equation}
Under the baseline production specification, this becomes
\begin{equation}
\Delta_A^c(\bar{K})
=
\frac{1}{6}
\lambda B i_U^{\alpha_U}(i_A^c)^{\alpha_A}.
\label{eq:cert-premium-closed}
\end{equation}

The delegation premium is a property-rights premium. It arises because certification changes the effective control regime from Platform Control to User Control. A higher $\lambda$ increases the value of the protected User--Agent outside option, while higher fixed or variable certification costs make certification less attractive.

\begin{proposition}[Threshold Investment Equilibrium]
\label{prop:threshold_investment}
Suppose the certification threshold is binding relative to the agent's unconstrained User-Control choice, so that $i_A^c(\bar{K})>i_A^{\Omega_U*}$. If the incentive-compatibility condition in equation~\eqref{eq:cert-agent-ic} holds, the agent chooses to certify. If, in addition, the certified payoff is locally concave and the unconstrained optimum lies below the threshold, the agent's optimal certified investment is
\[
i_A^{\mathrm{Cert}*}=i_A^c(\bar{K}).
\]
\end{proposition}

\begin{proof}
If the agent does not meet the threshold, its payoff is bounded above by $R_A^{P*}$. When equation~\eqref{eq:cert-agent-ic} holds, meeting the certification threshold weakly dominates remaining uncertified. If the threshold is binding and further increases in $i_A$ do not change the effective control regime, local concavity implies that the agent chooses the lowest investment that satisfies certification. Hence $i_A^{\mathrm{Cert}*}=i_A^c(\bar{K})$.
\end{proof}

The incentive constraint also shows when certification may fail. A low value of $\lambda$ reduces the access premium, while high fixed or variable compliance costs make certification less attractive. In those cases, the second-best rule may require additional instruments, such as access-fee adjustments, liability credits, compliance bonds, or targeted subsidies for verifiable risk mitigation.

\subsection{Participation and Authorization Constraints}
\label{subsec:cert-participation-authorization}

For Certified Delegation to be feasible, the main parties must prefer it to their relevant outside options. The platform's participation constraint can be written as
\begin{equation}
R_P^{\mathrm{Cert}}+T_P
\geq
R_P^{\Omega_P}-L_P,
\label{eq:cert-platform-pc}
\end{equation}
where $T_P$ denotes transfers or fees received by the platform, such as access fees or audit-related payments. The term $L_P$ denotes the cost of resisting automated delegation, including legal costs, regulatory exposure, reputational losses, or user churn.

The user's participation constraint is
\begin{equation}
R_U^{\mathrm{Cert}}-\ell_U(\bar{K})
\geq
R_U^{\mathrm{Manual}},
\label{eq:cert-user-pc}
\end{equation}
where $R_U^{\mathrm{Manual}}$ is the user's payoff from manual account use. The term $\ell_U(\bar{K})$ captures user-side frictions from certification, such as authorization steps, configuration costs, and monitoring effort.

Certification of the agent is not sufficient by itself. The proxy must also act under explicit user authorization, remain revocable, and stay within the scope of the user's account rights. These conditions keep delegation distinct from account transfer, independent scraping, unauthorized access, or unrestricted API access. A certified proxy loses protection when it exceeds the user's authorization, ignores revocation, violates rate limits, bypasses security systems, or exposes third-party data outside the authorized task scope.

\subsection{Welfare Optimization and the Second-Best Standard}
\label{subsec:cert-welfare}

Certified Delegation improves welfare over Platform Control when the gains from reducing hold-up and lowering risk exceed the costs of certification. Let $\chi(\bar{K})$ denote administrative, monitoring, and dispute-resolution costs. A sufficient condition for welfare improvement is
\begin{equation}
\Delta W_{\mathrm{Cert}}^{UA}
+
\Delta W_{\mathrm{Cert}}^{P}
+
\Delta W_{\mathrm{Cert}}^{K}
>
F_c
+
c_c(i_A^{\mathrm{Cert}})
+
\chi(\bar{K}),
\label{eq:cert-welfare-condition}
\end{equation}
where $\Delta W_{\mathrm{Cert}}^{UA}$ captures the recovery of user and agent investment, $\Delta W_{\mathrm{Cert}}^{P}$ captures the effect on platform investment, and $\Delta W_{\mathrm{Cert}}^{K}$ captures the change in expected risk.

Because $i_A$ is a composite capability variable, the welfare decomposition should be read at the level of the model rather than as a primitive engineering decomposition. In the scalar specification, the same investment variable supports both productive adaptation and compliance capability. The welfare condition therefore captures the joint benefit of access-protecting certification: it can restore delegation-related investment incentives while requiring the agent to satisfy a verifiable risk standard.

\begin{proposition}[Welfare-Optimal Certification]
\label{prop:welfare_optimal_cert}
Suppose there exists a non-empty set of thresholds $\mathcal{K}$ satisfying the welfare condition in equation~\eqref{eq:cert-welfare-condition} and the participation constraints in equations~\eqref{eq:cert-platform-pc} and~\eqref{eq:cert-user-pc}. Then a welfare-improving Certified Delegation regime exists. The second-best threshold solves
\begin{equation}
\bar{K}^* \in \arg\max_{\bar{K} \in \mathcal{K}} W(\Omega_{\mathrm{Cert}}(\bar{K})).
\label{eq:optimal-kbar}
\end{equation}
\end{proposition}

\begin{proof}
The regulator or certifier chooses $\bar{K}$ taking into account the agent's threshold response $i_A^c(\bar{K})$, the costs of compliance, and the participation constraints of the user and the platform. If the feasible set $\mathcal{K}$ is non-empty, maximizing welfare over this set gives the second-best standard. A welfare-improving certified regime exists whenever the maximized value over $\mathcal{K}$ exceeds the relevant polar-regime benchmark.
\end{proof}

The trade-off is straightforward. If the standard is too lenient, Certified Delegation approaches User Control and leaves too much residual risk. If the standard is too strict, compliance becomes costly and the agent may not certify, in which case the regime approaches Platform Control. The second-best threshold balances these two distortions.

\subsection{Imperfect Certification and Strategic Risk Margins}
\label{subsec:cert-imperfect}

The baseline mechanism assumes that risk $K(i_A)$ is observed without error. In practice, certification is noisy. Suppose the certifier observes
\begin{equation}
z=K(i_A)+\varepsilon,
\label{eq:cert-noisy-signal}
\end{equation}
where $\varepsilon$ is an independent error term with distribution function $F(\cdot)$ and density $f(\cdot)$. The certifier grants access protection if the probability of satisfying the risk standard is at least $p^* \in (0,1)$:
\begin{equation}
\Pr(z\leq \bar{K}\mid i_A)\geq p^*.
\label{eq:cert-probability-rule}
\end{equation}
This is equivalent to
\begin{equation}
K(i_A)\leq \bar{K}-F^{-1}(p^*).
\label{eq:cert-effective-threshold}
\end{equation}
Thus, when the certifier requires a high confidence level, the effective standard becomes stricter.

Let
\[
\pi(i_A)=F(\bar{K}-K(i_A))
\]
denote the probability that the agent obtains certification. The agent's expected payoff is
\begin{equation}
\operatorname{E}[R_A^{\mathrm{Cert}}]
=
\pi(i_A)
\left[
\varphi_A^{\Omega_U}(i_U,i_A,i_P)-F_c-c_c(i_A)
\right]
+
\left[
1-\pi(i_A)
\right]
\varphi_A^{\Omega_P}(i_U,i_A,i_P)
-
C_A(i_A).
\label{eq:cert-expected-payoff}
\end{equation}
The effect of investment on the probability of certification is
\begin{equation}
\pi'(i_A)
=
f(\bar{K}-K(i_A))
\left[-K'(i_A)\right]
>0.
\label{eq:cert-probability-derivative}
\end{equation}

Investment therefore raises the chance of receiving access protection. With noisy certification, however, false positives and false negatives are unavoidable. False positives allow unsafe agents to obtain protected access and therefore call for monitoring, liability, and suspension rules after certification. False negatives deny protection to qualified agents and therefore call for appeal, recertification, and review procedures. These institutional safeguards are necessary because certification changes legal-economic access status rather than merely assigning a technical score.

\subsection{Screening Heterogeneous Agent Populations}
\label{subsec:cert-screening}

Certified Delegation can also screen heterogeneous agents. Let the agent's type be $\theta>0$, where higher $\theta$ means lower marginal cost of capability or compliance. The agent's cost is
\begin{equation}
C_A(i_A;\theta)
=
\frac{\kappa_A}{2\theta}i_A^2.
\label{eq:cert-type-cost}
\end{equation}
For a type-$\theta$ agent, define the net gain from certification as
\begin{equation}
D(\theta;\bar{K})
=
\left[
\varphi_A^{\Omega_U}(i_U,i_A^c,i_P)
-
C_A(i_A^c;\theta)
-
F_c
-
c_c(i_A^c)
\right]
-
R_A^{P*}(\theta).
\label{eq:cert-screening-gain}
\end{equation}

\begin{proposition}[Adverse Selection Sorting]
\label{prop:screening_sorting}
If $D(\theta;\bar{K})$ is strictly increasing in $\theta$, there exists a cutoff type $\theta^*(\bar{K})$ such that the agent certifies if and only if
\[
\theta\geq \theta^*(\bar{K}).
\]
\end{proposition}

\begin{proof}
Strict monotonicity of $D(\theta;\bar{K})$ implies that the net gain from certification is higher for more efficient agents. If there is an interior type at which $D(\theta;\bar{K})=0$, the Intermediate Value Theorem gives a cutoff $\theta^*(\bar{K})$. Agents above the cutoff certify, while agents below it do not.
\end{proof}

This sorting result is useful only if the certification standard is tied to genuine risk and compliance requirements. If the standard is made unnecessarily costly, opaque, discriminatory, or platform-controlled without review, certification can become an entry barrier rather than a safety screen. The access-protecting character of certification therefore requires procedural safeguards: published standards, non-discriminatory review, evidence-based refusal, appeal rights, and periodic updating.

\subsection{Summary of Mechanism Dynamics}
\label{subsec:cert-summary}

Certified Delegation works through three channels. First, it reduces hold-up by giving compliant agents a predictable access path. Second, it limits risk by conditioning access protection on a verifiable standard. Third, it screens agents by making certification more attractive to providers with lower compliance costs.

The key point is that certification changes both risk and rights. It lowers expected harm by imposing compliance requirements, but it also reallocates residual control over the account interface. A certified, user-authorized, scope-limited, and revocable proxy receives protection against exclusion merely on the ground that execution is automated. An uncertified or non-compliant proxy remains subject to platform refusal. The mechanism improves welfare when the gains from restoring delegation-related investment and bounding residual risk exceed the costs of certification, monitoring, and enforcement.

\section{Illustrative Numerical Analysis and Counterfactual Simulations}
\label{sec:calibration}

This section provides an illustrative numerical analysis of the model. The purpose is not to estimate the welfare effects of any actual dispute. The purpose is to illustrate how the theoretical channels operate under transparent parameter choices. The simulations should therefore be read as mechanism illustrations rather than empirical or structural estimates.

The numerical exercises have three roles. First, they show how Platform Control, User Control, and Certified Delegation differ in equilibrium investment, risk, welfare, and deadweight loss under a maintained parameter profile. Second, they illustrate how the preferred polar regime can vary across environments depending on the relative importance of user-agent contribution, platform contribution, outside-option effectiveness, and baseline risk. Third, they show how a certification-contingent access rule can improve welfare by restoring delegation-related investment while bounding residual risk.

As discussed in Section~\ref{subsec:reduced_form_agent_investment}, the agent's investment $i_A$ is a reduced-form composite. It includes both productive capability, such as platform-specific adaptation and workflow execution, and compliance capability, such as authorization controls, auditability, rate-limit adherence, and data minimization. The numerical analysis therefore cannot separately identify productive engineering investment from safety-compliance investment at the level of primitive choices. The results should be interpreted as illustrating the joint mechanism generated by access-protecting certification.

\subsection{Maintained Parameter Profile}
\label{subsec:calib-design}

The baseline parameter profile satisfies decreasing returns to scale:
\begin{equation}
\alpha_U+\alpha_A+\alpha_P=0.75<1.
\label{eq:calib-decreasing-returns}
\end{equation}
Table~\ref{tab:calib-baseline} reports the maintained parameter values used in the baseline numerical exercise.

\begin{table}[htbp]
\centering
\small
\caption{Maintained Parameter Profile}
\label{tab:calib-baseline}
\begin{tabular}{llc}
\toprule
Symbol & Parameter Description & Maintained Value \\
\midrule
$\alpha_U$ & User output elasticity & 0.20 \\
$\alpha_A$ & AI Agent output elasticity & 0.30 \\
$\alpha_P$ & Platform output elasticity & 0.25 \\
$\kappa_U$ & User investment cost coefficient & 1.00 \\
$\kappa_A$ & Agent investment cost coefficient & 1.20 \\
$\kappa_P$ & Platform investment cost coefficient & 0.80 \\
$\rho_0$ & Baseline system vulnerability & 0.15 \\
$\rho_1$ & Risk-mitigation technological efficiency & 2.00 \\
$\beta_{UP}$ & User--Platform sub-coalition factor & 0.60 \\
$\beta_{AP}$ & Agent--Platform sub-coalition factor & 0.30 \\
$\lambda$ & Non-cooperative workaround efficiency & 0.25 \\
$B$ & Global production scale parameter & 10.0 \\
\bottomrule
\end{tabular}
\end{table}

The parameter values are not estimated from a specific platform environment. They are maintained values chosen to make the model's mechanisms transparent. The table therefore defines the numerical environment used for illustration; it should not be interpreted as an empirical calibration of any particular platform, agent provider, or dispute.

\subsection{Baseline Mechanism Comparison}
\label{subsec:baseline-equilibria}

Table~\ref{tab:calib-equilibria} reports the equilibrium investment vectors, gross output, expected risk costs, social welfare, and deadweight losses under four allocations: the first-best benchmark, Platform Control ($\Omega_P$), User Control ($\Omega_U$), and Certified Delegation evaluated at the welfare-maximizing standard $\bar{K}^*$.

\begin{table}[htbp]
\centering
\small
\setlength{\tabcolsep}{4pt} 
\caption{Baseline Mechanism Comparison}
\label{tab:calib-equilibria}
\begin{tabularx}{\textwidth}{l *{4}{>{\centering\arraybackslash}X}}
\toprule
\textbf{Economic Metric} & \textbf{First-Best} & \textbf{Platform Control} ($\Omega_P$) & \textbf{User Control} ($\Omega_U$) & \textbf{Certified Delegation} ($\Omega_{\mathrm{Cert}}(\bar{K}^*)$) \\
\midrule
User Investment ($i_U^*$) & 1.75 & 0.88 & 0.92 & 1.01 \\
\addlinespace
Agent Investment ($i_A^*$) & 1.89 & 1.01 & 1.08 & 1.96 \\
\addlinespace
Platform Investment ($i_P^*$) & 2.35 & 1.16 & 1.18 & 1.30 \\
\addlinespace
Gross Output ($V^*$) & 12.39 & 9.46 & 9.65 & 11.52 \\
\addlinespace
Expected Risk Cost ($K^*$) & 0.020 & 0.043 & 0.041 & 0.020 \\
\addlinespace
Aggregate Welfare ($W^*$) & 10.90 & 8.78 & 8.99 & 9.86 \\
\addlinespace
Deadweight Loss (\% of $W^{FB}$) & --- & 19.5\% & 17.5\% & 9.5\% \\
\bottomrule
\end{tabularx}
\end{table}

Table~\ref{tab:calib-equilibria} should not be read as an empirical estimate of any platform environment. It shows that, under the maintained parameter profile, Platform Control mainly loses welfare through depressed User--Agent investment, while Certified Delegation recovers part of that investment by conditioning access protection on a risk threshold. In this numerical environment, agent investment falls to $i_A^*(\Omega_P)=1.01$ under Platform Control, compared with the first-best value $i_A^{FB}=1.89$. User Control raises agent investment only slightly, to $i_A^*(\Omega_U)=1.08$. Certified Delegation generates a larger change because the threshold induces the agent to meet the risk standard. Agent investment rises to $i_A^*(\Omega_{\mathrm{Cert}})=1.96$, and deadweight loss falls from 19.5\% under Platform Control to 9.5\% under Certified Delegation.

\begin{figure}[htbp]
\centering
\includegraphics[width=\textwidth]{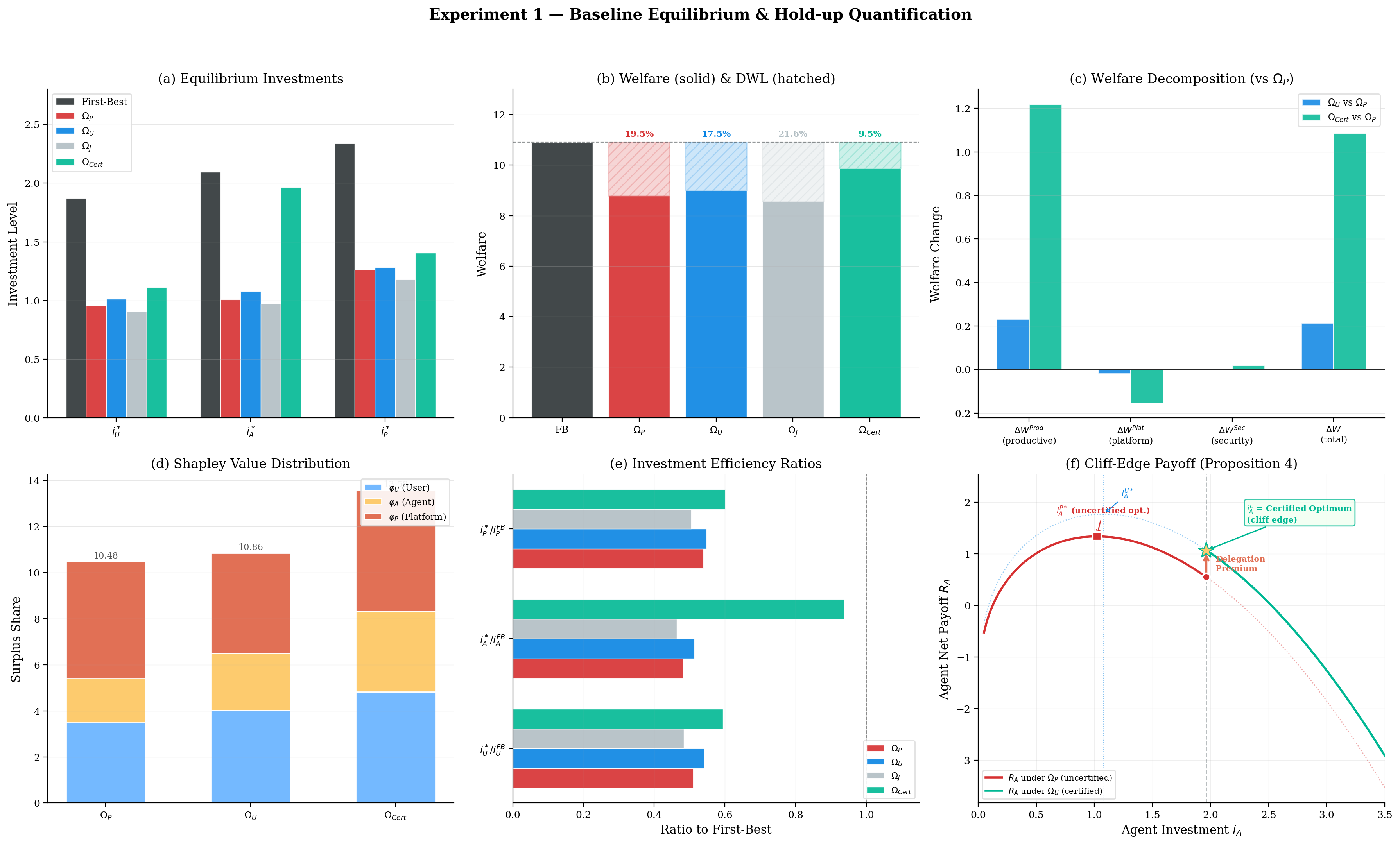}
\caption{Baseline diagnostics across control regimes.}
\label{fig:calib-baseline-diagnostics}
\end{figure}

Figure~\ref{fig:calib-baseline-diagnostics} summarizes the investment, welfare, bargaining, compliance-cost, and risk-frontier diagnostics for the maintained parameter profile. The figure illustrates the same mechanism reported in Table~\ref{tab:calib-equilibria}; it does not establish a global ranking of regimes outside the parameter region considered.

\subsection{Regime Map in Elasticity Space}
\label{subsec:contribution-regime-map}

We next vary the agent and platform output elasticities, $(\alpha_A,\alpha_P)$, to illustrate how the preferred polar regime changes across environments. For this exercise only, we introduce a simple moral-hazard extension: under User Control, only a fraction $\mu=0.15$ of agent investment is directed toward security compliance. We also use a higher baseline risk value, $\rho_0=0.50$. This extension is used only for the regime-map visualization and is not used in the baseline mechanism comparison or the case counterfactuals.

The grid contains $91\times 91$ points. In this numerical domain, User Control is welfare-dominant at 3,737 points, or about 54\% of the grid. Platform Control is welfare-dominant at 3,166 points, or about 46\% of the grid. The estimated indifference boundary slopes upward: as the platform elasticity $\alpha_P$ increases, a higher agent elasticity $\alpha_A$ is needed for User Control to remain welfare-superior.

\begin{figure}[htbp]
\centering
\includegraphics[width=0.88\textwidth]{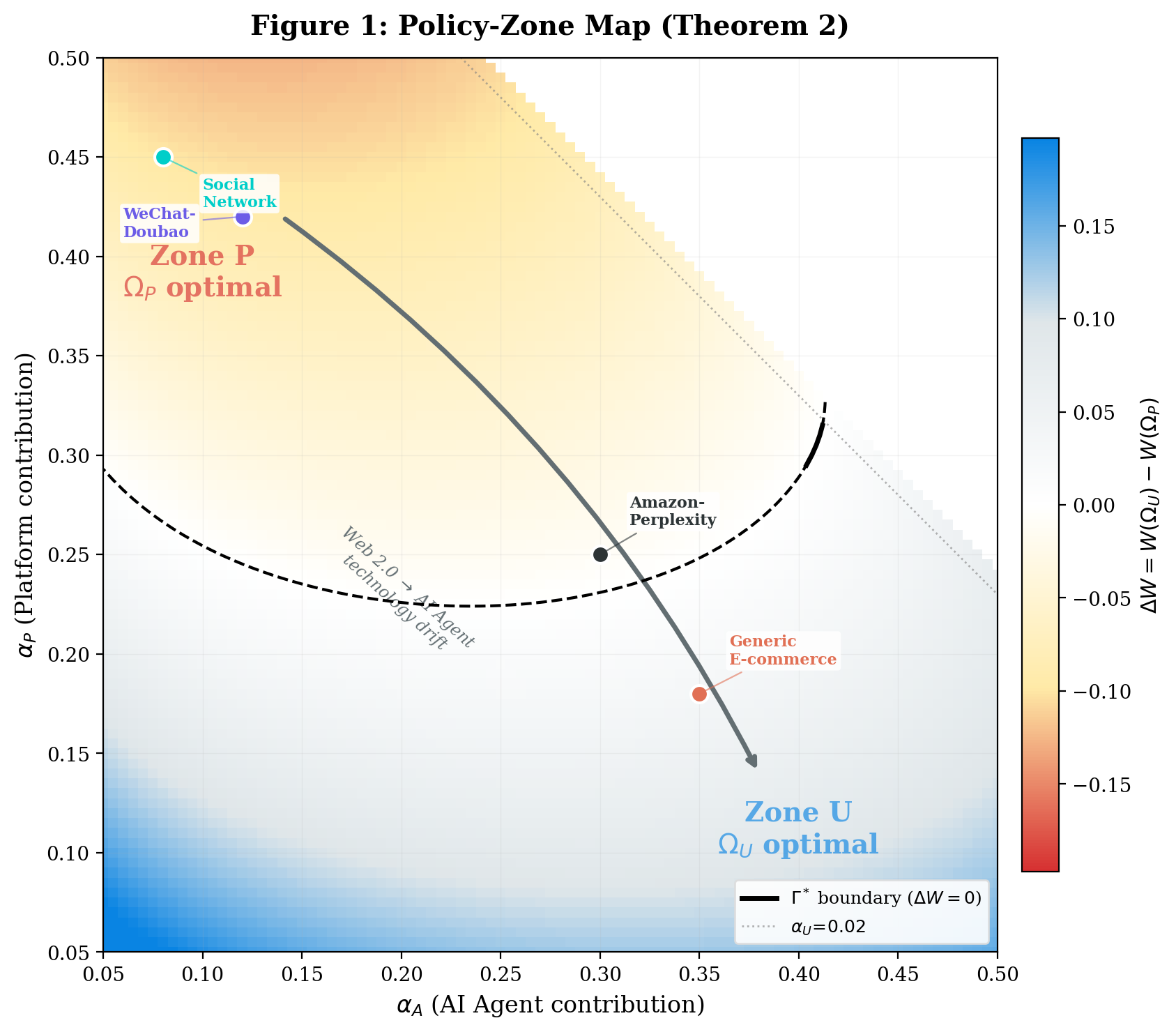}
\caption{Regime map in elasticity space.}
\label{fig:calib-regime-map}
\end{figure}

Figure~\ref{fig:calib-regime-map} should be read as a parametric illustration. It shows how the welfare-dominant polar regime varies across the $(\alpha_A,\alpha_P)$ grid under the specified extension. It does not imply a universal closed-form threshold, and the location of the boundary depends on the maintained functional forms and parameter values.

\subsection{Welfare Properties of the Certification Threshold}
\label{subsec:certification-threshold}

We now examine the Certified Delegation threshold. The risk standard is varied over
\begin{equation}
\bar{K}\in[0.003,0.149],
\label{eq:calib-threshold-range}
\end{equation}
using 500 grid points. Welfare under $\Omega_{\mathrm{Cert}}(\bar{K})$ follows an inverted-U pattern in the maintained numerical environment. The welfare-maximizing standard is
\begin{equation}
\bar{K}^*=0.0305,
\qquad
W^*=9.86.
\label{eq:calib-threshold-optimum}
\end{equation}
For comparison,
\begin{equation}
W(\Omega_P)=8.78,
\qquad
W(\Omega_U)=8.99.
\label{eq:calib-polar-welfare}
\end{equation}
Certified Delegation outperforms both polar regimes for thresholds in the interval
\begin{equation}
[\underline{\bar{K}},\overline{\bar{K}}]
=
[0.022,0.047].
\label{eq:calib-dominance-interval}
\end{equation}
This interval has width 0.025, or about 17\% of the threshold domain considered here.

\begin{figure}[htbp]
\centering
\includegraphics[width=0.88\textwidth]{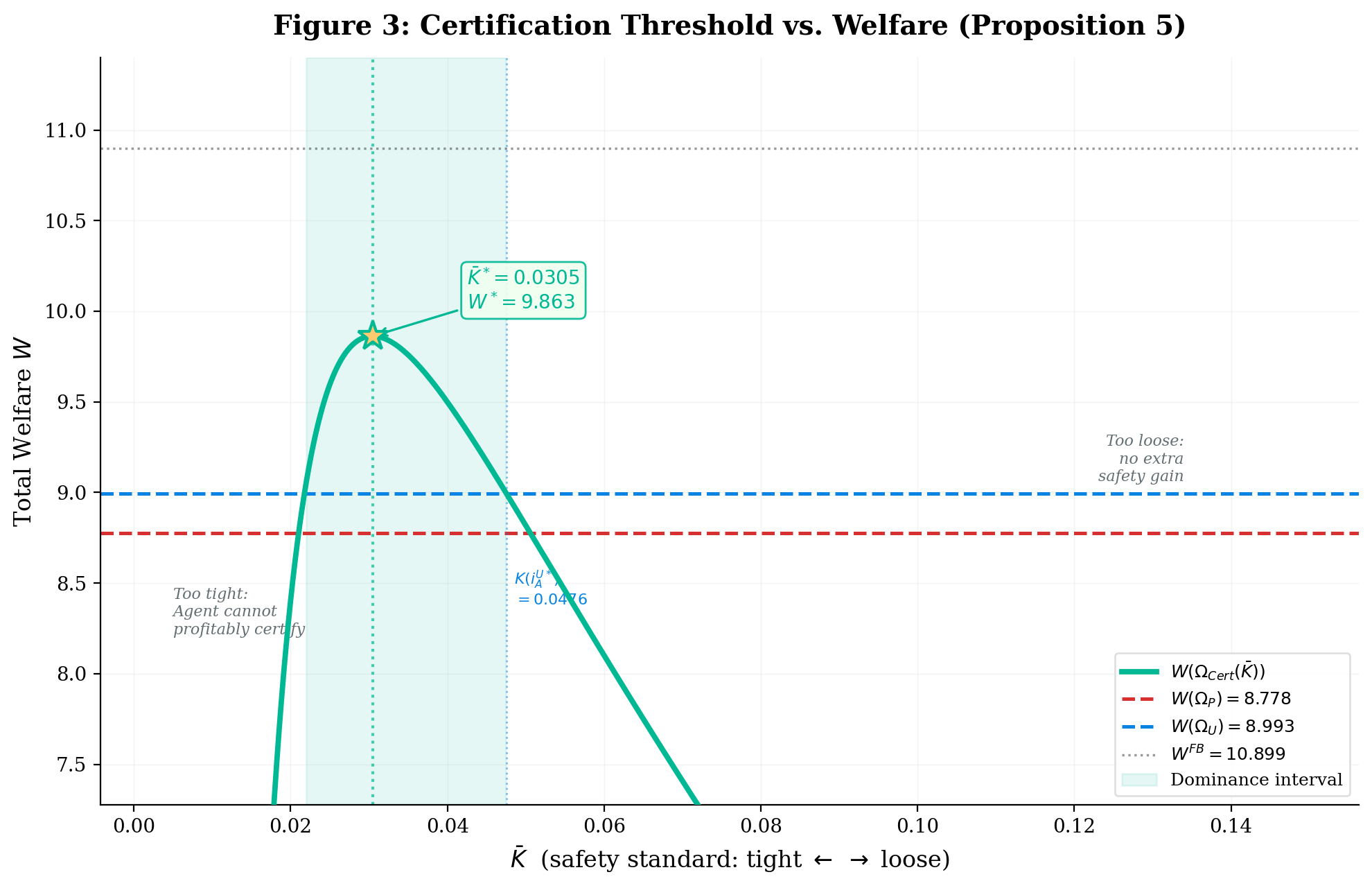}
\caption{Welfare under Certified Delegation as the risk standard varies.}
\label{fig:calib-cert-threshold}
\end{figure}

Figure~\ref{fig:calib-cert-threshold} illustrates the trade-off in the certification standard. If the standard is too lenient, Certified Delegation approaches User Control and does little to reduce residual risk. If the standard is too strict, the agent must invest heavily to qualify, raising compliance costs and potentially violating the incentive-compatibility constraint. The inverted-U shape is a property of the maintained numerical environment rather than a global theorem.

In the maintained profile, the welfare-maximizing threshold $\bar{K}^*=0.0305$ requires $i_A^c=1.96$. This value is above the agent's self-enforcing boundary in the private incentive problem, suggesting that implementation may require additional instruments such as access-fee adjustments, liability credits, compliance bonds, or targeted subsidies for verifiable risk mitigation.

\subsection{Scale-Normalized Counterfactual I: AI-Assisted Commerce}
\label{subsec:case-ai-commerce}

The first stylized counterfactual represents an e-commerce setting in which automated search, price comparison, and product matching generate substantial user-agent surplus. We use the Amazon--Perplexity environment only as an institutional archetype. The case label is not a legal or empirical finding about the named dispute.

The parameter profile is
\begin{equation}
\alpha_A=0.32,
\qquad
\alpha_P=0.25,
\qquad
\rho_0=0.12,
\qquad
\lambda=0.20,
\label{eq:amazon-parameters}
\end{equation}
together with an agent cost coefficient $\kappa_A=1.50$ and an annualized value anchor of \$50B/yr. The dollar values below are scale-normalized illustrations, not estimates of the welfare effects of the actual dispute.

Under these parameters, the environment falls in the User Control region. The welfare difference between User Control and Platform Control is
\begin{equation}
\Delta W(\Omega_U-\Omega_P)=+\$0.52\text{B/yr}.
\label{eq:amazon-zone-gain}
\end{equation}
Certified Delegation generates an additional scale-normalized welfare gain of \$2.80B/yr relative to the better polar regime. Deadweight loss falls from 20.2\% under Platform Control to 9.6\% under Certified Delegation. The main channel is investment recovery: relative agent investment increases from $i_A^*/i_A^{FB}=47.3\%$ under Platform Control to $i_A^c/i_A^{FB}=93.7\%$ at the certification threshold. In the decomposition, the productive investment channel contributes \$3.73B, while the direct risk-reduction channel contributes \$0.05B.

\begin{figure}[htbp]
\centering
\includegraphics[width=0.92\textwidth]{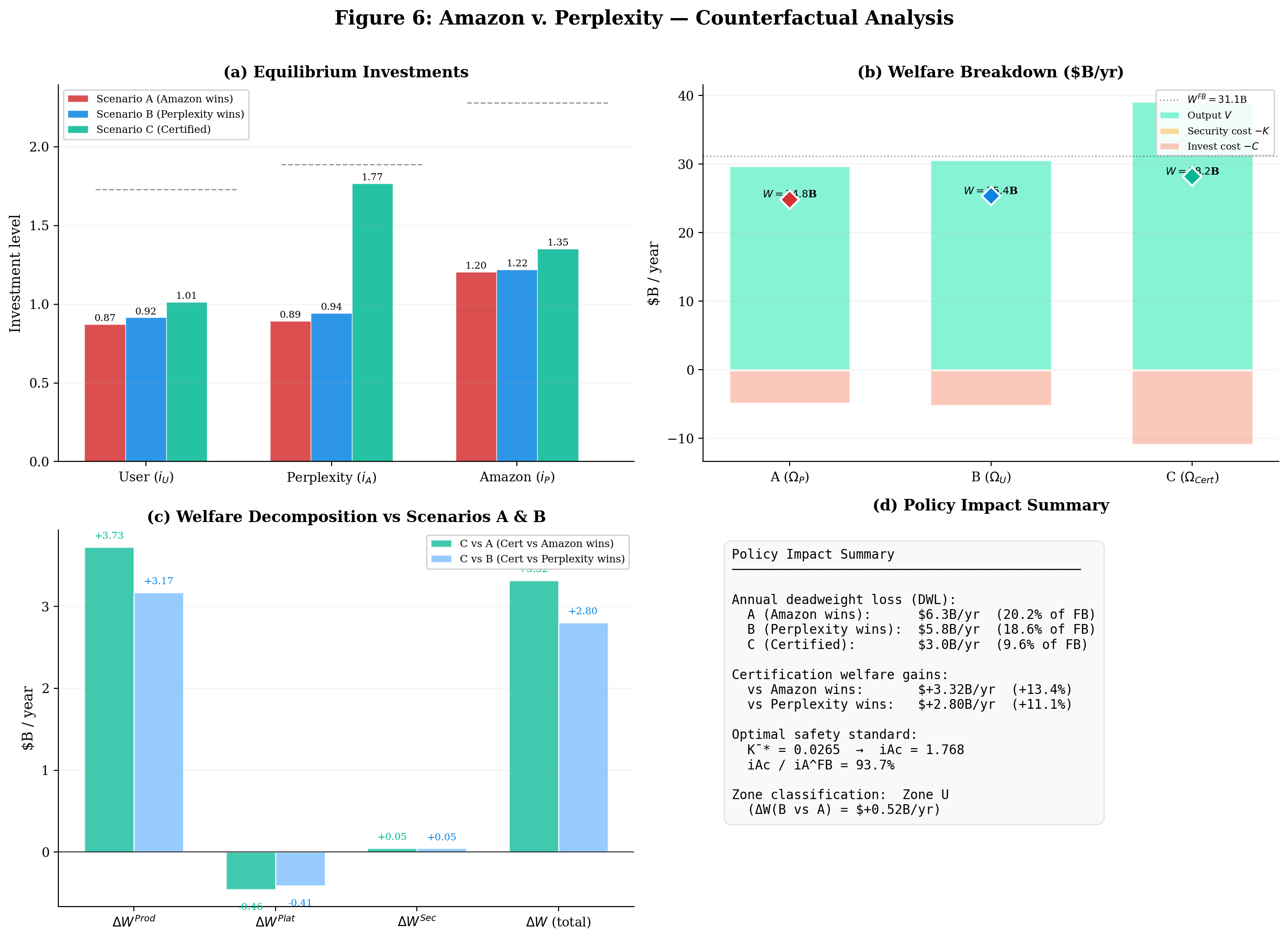}
\caption{Scale-normalized counterfactual: AI-assisted commerce.}
\label{fig:calib-amazon}
\end{figure}

Figure~\ref{fig:calib-amazon} illustrates the mechanism in a commerce-like environment. The case label is used to discipline parameter interpretation; it should not be read as a factual claim about the legal merits, business facts, or actual welfare consequences of any specific dispute.

\subsection{Scale-Normalized Counterfactual II: Closed Social Ecosystem}
\label{subsec:case-closed-social}

The second stylized counterfactual represents a closed social ecosystem in which identity systems, privacy boundaries, and platform infrastructure play a larger role. We use the WeChat--Doubao environment only as an institutional archetype. The case label is not a legal or empirical finding about the named dispute.

The parameter profile is
\begin{equation}
\alpha_A=0.18,
\qquad
\alpha_P=0.42,
\qquad
\rho_0=0.22,
\qquad
\lambda=0.08,
\label{eq:wechat-parameters}
\end{equation}
together with a higher User--Platform sub-coalition factor, $\beta_{UP}=0.75$, and an annualized value anchor of \$30B/yr. The dollar values below are scale-normalized illustrations, not estimates of the welfare effects of the actual dispute.

Under these parameters, the environment lies close to the boundary between the two polar regimes:
\begin{equation}
\Delta W(\Omega_U-\Omega_P)=+\$0.07\text{B/yr}.
\label{eq:wechat-zone-gain}
\end{equation}
Certified Delegation increases scale-normalized welfare by \$1.05B/yr relative to the better polar alternative. Deadweight loss falls from 17.4\% under Platform Control to 11.4\% under Certified Delegation. As in the commerce case, the main effect comes through investment recovery. Relative agent investment rises from $i_A^*/i_A^{FB}=48.2\%$ under Platform Control to $i_A^c/i_A^{FB}=92.0\%$ under certification. The productive channel contributes \$1.32B, while the direct risk-mitigation channel contributes \$0.05B.

\begin{figure}[htbp]
\centering
\includegraphics[width=0.92\textwidth]{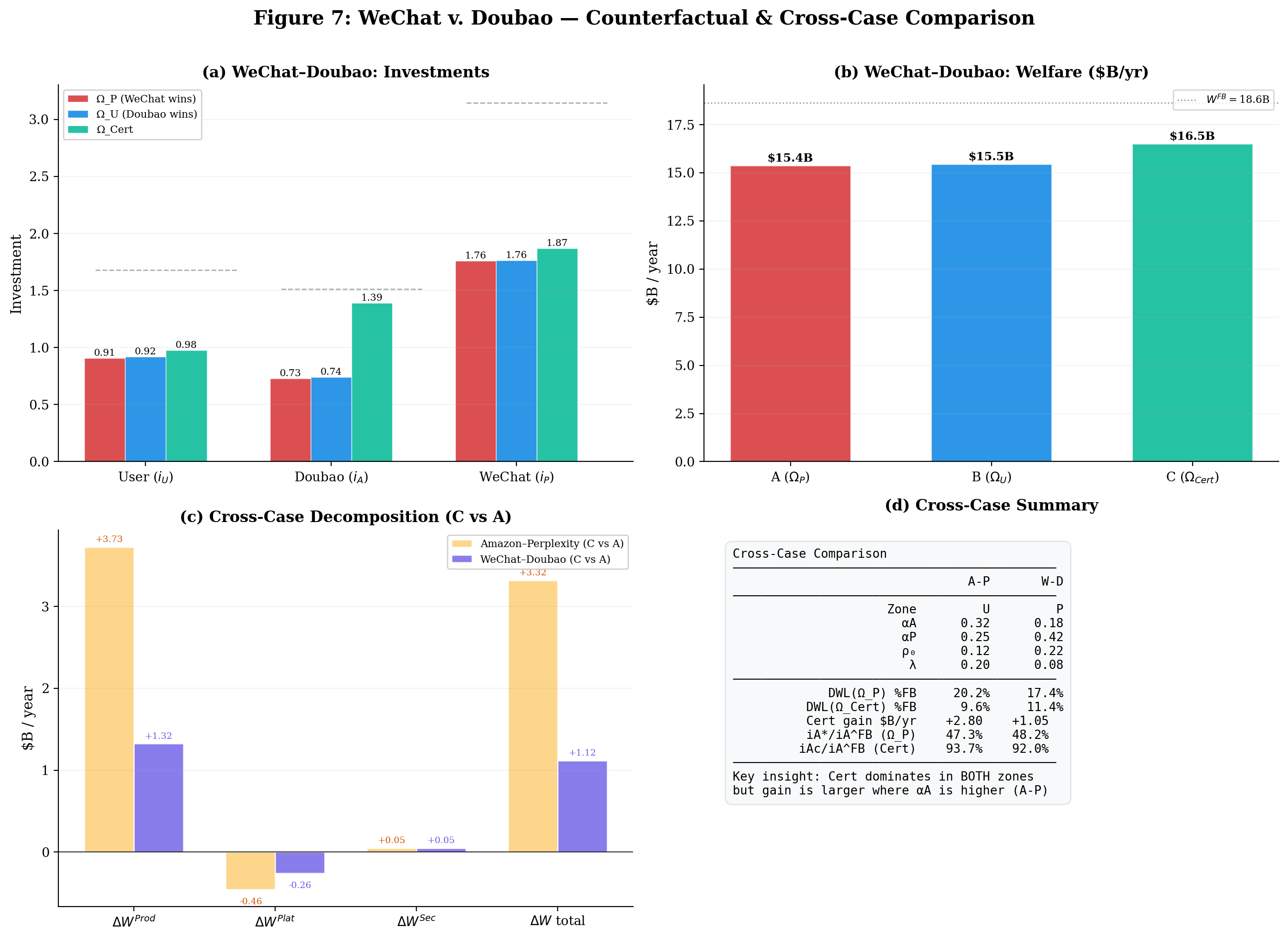}
\caption{Scale-normalized counterfactual: closed social ecosystem.}
\label{fig:calib-wechat}
\end{figure}

Figure~\ref{fig:calib-wechat} illustrates the mechanism in a social-ecosystem-like environment. The case label is used to discipline parameter interpretation; it should not be read as a factual claim about the legal merits, business facts, or actual welfare consequences of any specific dispute.

\subsection{Comparative Counterfactual Synthesis}
\label{subsec:case-comparison}

Table~\ref{tab:case-counterfactuals} compares the two stylized counterfactual environments. The dollar-denominated entries are scale-normalized illustrations based on the stated annualized value anchors. They are not estimates of the welfare effects of the actual disputes.

\begin{table}[htbp]
\centering
\small
\setlength{\tabcolsep}{5pt}
\caption{Scale-Normalized Counterfactual Simulations}
\label{tab:case-counterfactuals}
\begin{tabularx}{\textwidth}{Xcc}
\toprule
\textbf{Economic Metric} & \textbf{AI-Assisted Commerce} & \textbf{Closed Social Ecosystem} \\
\midrule
Institutional Archetype & Amazon--Perplexity Environment & WeChat--Doubao Environment \\
\addlinespace
Annualized Value Anchor & \$50B/yr & \$30B/yr \\
\addlinespace
Production Elasticities $(\alpha_A, \alpha_P)$ & $(0.32, 0.25)$ & $(0.18, 0.42)$ \\
\addlinespace
Primitive Constraints $(\rho_0, \lambda)$ & $(0.12, 0.20)$ & $(0.22, 0.08)$ \\
\addlinespace
Welfare-Dominant Polar Zone & Zone $\Omega_U$ & Near boundary \\
\addlinespace
Polar Welfare Differential $\Delta W(\Omega_U-\Omega_P)$ & \$0.52B/yr & \$0.07B/yr \\
\addlinespace
Net Welfare Gain of $\Omega_{\mathrm{Cert}}$ vs. Polar Baseline & \$2.80B/yr & \$1.05B/yr \\
\addlinespace
Deadweight Loss under Platform Control & 20.2\% & 17.4\% \\
\addlinespace
Deadweight Loss under Certified Delegation & 9.6\% & 11.4\% \\
\addlinespace
Relative Proxy Investment under Platform Control \newline ($i_A^*/i_A^{FB}$) & 47.3\% & 48.2\% \\
\addlinespace
Relative Proxy Investment under Certified Delegation \newline ($i_A^c/i_A^{FB}$) & 93.7\% & 92.0\% \\
\addlinespace
Productive Channel Capital Recovery & \$3.73B & \$1.32B \\
\addlinespace
Direct Security Channel Welfare Yield & \$0.05B & \$0.05B \\
\bottomrule
\end{tabularx}
\end{table}

The comparison points to a common pattern. In both stylized environments, Certified Delegation improves welfare mainly by reducing hold-up and restoring agent investment toward the first-best benchmark. The direct risk-reduction channel is positive but small in these parameterizations. Thus, in these numerical exercises, the safety standard matters not only because it lowers risk directly, but also because it creates a credible condition under which platform exclusion is limited.

\subsection{Robustness Analysis and Interpretive Boundaries}
\label{subsec:robustness-boundaries}

We finally examine how the main qualitative predictions behave across alternative specifications. The robustness exercise considers 14 model variants and tracks four predictions: hold-up ordering, zone partitioning, certification dominance, and the direction of the optimal threshold response to baseline risk.

For transparency, Table~\ref{tab:robustness-variants} reports the variant definitions used in the robustness exercise. The entries should correspond exactly to the simulation protocol. If the simulation code uses different variant names or parameter perturbations, this table should be updated to match the code before submission.

\begin{table}[htbp]
\centering
\small
\caption{Robustness Variant Definitions}
\label{tab:robustness-variants}
\begin{tabularx}{\textwidth}{lX}
\toprule
\textbf{Variant} & \textbf{Change Relative to Baseline} \\
\midrule
V1 & Higher baseline risk, $\rho_0$. \\
V2 & Lower outside-option effectiveness, $\lambda$. \\
V3 & Higher agent investment cost, $\kappa_A$. \\
V4 & Lower Agent--Platform sub-coalition factor, $\beta_{AP}$. \\
V5 & Higher User--Platform sub-coalition factor, $\beta_{UP}$. \\
V6 & Higher platform investment cost, $\kappa_P$. \\
V7 & Lower risk-mitigation efficiency, $\rho_1$. \\
V8 & Higher risk-mitigation efficiency, $\rho_1$. \\
V9 & Higher platform output elasticity, $\alpha_P$. \\
V10 & Higher agent output elasticity, $\alpha_A$. \\
V11 & Lower workaround effectiveness combined with higher baseline risk. \\
V12 & Higher certification cost or compliance friction. \\
V13 & Alternative sub-coalition values, $(\beta_{UP},\beta_{AP})$. \\
V14 & Alternative production-elasticity profile preserving decreasing returns. \\
\bottomrule
\end{tabularx}
\end{table}

Table~\ref{tab:robustness-variants} is a reporting device rather than a new experiment. Its role is to make the robustness exercise auditable. The numerical results reported below are unchanged; the table should simply mirror the variants already used in the simulation code.

Table~\ref{tab:robustness-summary} summarizes the qualitative support for the four predictions across the 14 variants.

\begin{table}[htbp]
\centering
\small
\setlength{\tabcolsep}{6pt}
\caption{Qualitative Robustness Summary}
\label{tab:robustness-summary}
\begin{tabularx}{\textwidth}{lYc}
\toprule
\textbf{Core Prediction} & \textbf{Economic Interpretation} & \textbf{Support Across Variants} \\
\midrule
Q1: Hold-up Ordering 
& Agent investment is lower under Platform Control:
\newline $i_A^*(\Omega_P)<i_A^*(\Omega_U)$ 
& 14/14 \\
\addlinespace
Q2: Zone Partitioning 
& The parameter space contains stable regions favoring either User Control or Platform Control 
& 11/14 \\
\addlinespace
Q3: Certification Dominance 
& The best Certified Delegation regime outperforms both polar regimes:
\newline $W(\Omega_{\mathrm{Cert}})>\max\{W(\Omega_P),W(\Omega_U)\}$ 
& 13/14 \\
\addlinespace
Q4: Threshold Direction 
& Higher baseline risk tightens the optimal safety standard:
\newline $\partial \bar{K}^*/\partial \rho_0 < 0$ 
& 11/14 \\
\bottomrule
\end{tabularx}
\end{table}

Table~\ref{tab:robustness-summary} should be interpreted qualitatively. The hold-up result is stable across all 14 variants. Certification dominance appears in most variants. By contrast, the exact partition of the parameter space and the comparative statics of the optimal threshold are less stable. This is consistent with the theoretical argument: the model supports a conditional mechanism, not a universal quantitative boundary.

\begin{figure}[htbp]
\centering
\includegraphics[width=\textwidth]{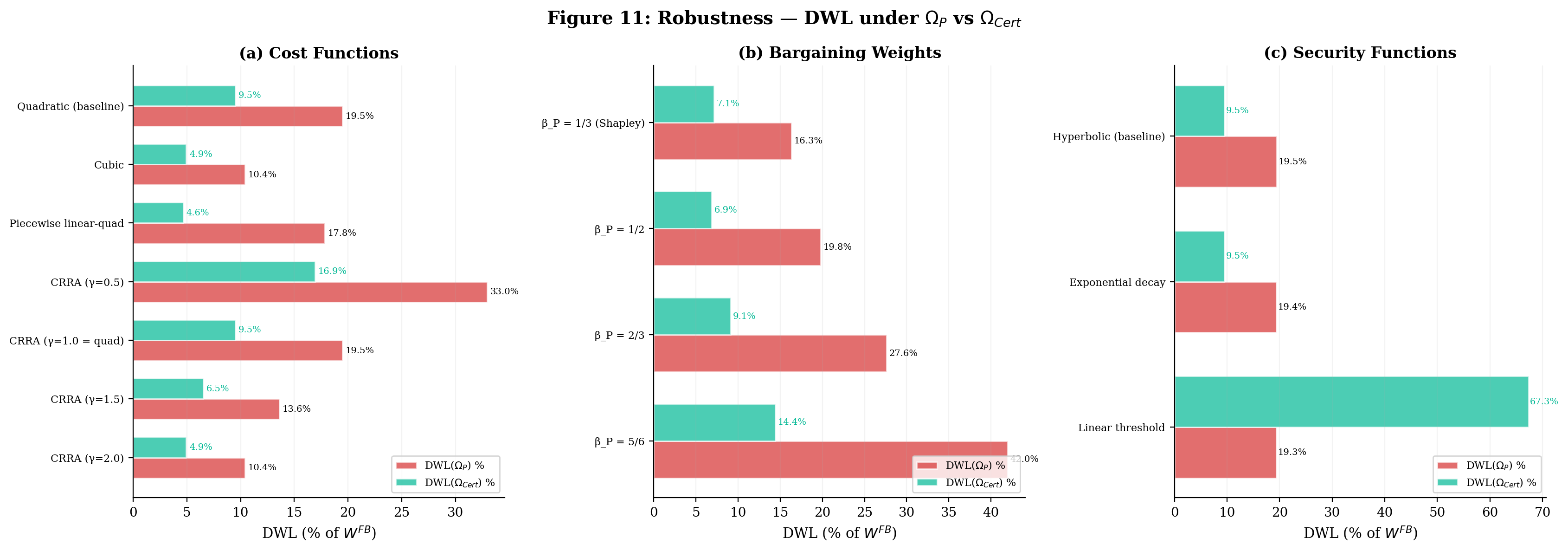}
\caption{Robustness diagnostic surface.}
\label{fig:calib-robustness}
\end{figure}

Figure~\ref{fig:calib-robustness} visualizes the robustness comparison across model variants. The figure supports the qualitative mechanism that Certified Delegation can reduce deadweight loss relative to Platform Control in many parameter environments. It does not imply that the same threshold, welfare ranking, or comparative static applies globally.

The robustness exercise therefore gives a mixed but informative picture. The strongest prediction is the hold-up ordering: Platform Control depresses agent investment relative to User Control across all variants considered. Certification dominance is also robust in most cases. The weaker predictions are the exact location of the polar-regime boundary and the direction of the optimal threshold response in every environment. These findings support the mechanism of the paper while limiting the scope of the numerical claims.

\subsection{Interpretive Boundaries}
\label{subsec:numerical-interpretive-boundaries}

The numerical analysis should be interpreted subject to four boundaries.

First, the parameter values are maintained values rather than estimated primitives. They are chosen to illustrate the mechanism, not to recover structural features of a particular platform market.

Second, the case labels are institutional archetypes. They discipline the interpretation of parameter profiles but do not constitute legal, factual, or empirical claims about the named disputes.

Third, dollar-denominated counterfactuals are scale-normalized illustrations based on stated annualized value anchors. They should not be read as welfare estimates for actual firms, platforms, users, or disputes.

Fourth, the scalar investment variable $i_A$ combines productive adaptation and compliance capability. The simulations therefore illustrate the joint effect of access-protecting certification rather than separately estimating production and safety investment technologies.

Subject to these boundaries, the numerical results help clarify the paper's main mechanism. Platform Control can reduce welfare by weakening the User--Agent coalition's investment incentives. User Control can mitigate hold-up but may leave residual risk. Certified Delegation can improve welfare when a verifiable certification threshold restores delegation-related investment while preserving the platform's right to refuse unsafe or non-compliant automation.

\section{Policy Implications and Conclusion}
\label{sec:policy_conclusion}

The model directly implies a conditional-access principle. Platform Control can protect infrastructure and give platforms room to respond to security risks, but it also gives the platform a strong ex-post veto over automated account operation. That veto can discourage users and agent providers from making platform-specific investments. User Control reduces this hold-up problem, but it may leave privacy, identity, security, congestion, and third-party risks insufficiently internalized.

The institutional details below translate this conditional-access principle into possible governance design features. They should not be read as uniquely optimal rules derived by the model. The model identifies the central trade-off: protected delegation can restore investment incentives, while certification can limit residual risk. The design question is how to make that protection conditional, reviewable, and proportionate.

We use \emph{Certified Delegation} to refer to the formal regime analyzed in the model. We use \emph{conditional delegation} to describe the broader policy principle. A user-authorized agent should receive protection against exclusion only when it satisfies verifiable requirements for authorization, revocability, auditability, risk control, and accountability. A certified agent would receive a protected route to the account interface. An uncertified, non-compliant, or high-risk agent would remain subject to the platform's right of refusal.

\subsection{The Conditional-Access Principle}
\label{subsec:policy-conditional-delegation}

Conditional delegation separates two questions that are often conflated. The first question is whether an account holder has a protectable interest in choosing an automated proxy to exercise rights that the user already holds. The second question is whether the selected proxy satisfies the safety and accountability requirements needed to protect the platform, non-delegating users, counterparties, and affected third parties.

This distinction is central to the policy implication. A qualified delegation right does not require platforms to give up ownership of their servers, code, identity systems, ranking tools, payment systems, security architecture, or core governance infrastructure. It also does not give users an unrestricted right to introduce any automated system into a platform environment. Instead, it limits arbitrary exclusion when the proxy is authorized by the user, remains within the user's account rights, and satisfies a defined certification standard.

The standard should be risk-proportionate. Low-risk tasks, such as product comparison, basic summarization, or preference-based search assistance, may require relatively light validation. Higher-risk tasks, such as social messaging, financial execution, healthcare management, legal submission, identity routing, or actions affecting third parties in material ways, require stronger authorization, logging, monitoring, and liability rules. The relevant policy question is therefore not whether automation is present. It is whether the automated proxy is bounded, observable, revocable, and accountable.

Table~\ref{tab:policy_model_to_design} summarizes how the model's mechanisms translate into policy design principles.

\begin{table}[htbp]
\centering
\small
\caption{From Model Mechanism to Policy Design}
\label{tab:policy_model_to_design}
\begin{tabularx}{\textwidth}{p{0.25\textwidth} X X}
\toprule
\textbf{Model Mechanism} & \textbf{Policy Principle} & \textbf{Design Implication} \\
\midrule
Platform veto can create hold-up
& Limit arbitrary exclusion of qualified user-authorized proxies
& Certified agents receive access protection; refusal must be evidence-based and reviewable. \\
\addlinespace
User Control can leave residual externalities
& Make delegation conditional on verifiable risk controls
& Certification should include authorization, revocation, auditability, rate limits, data minimization, and incident response. \\
\addlinespace
Certification changes residual control
& Treat certification as access-protecting certification, not merely technical scoring
& Once the proxy satisfies the standard, the platform cannot refuse access solely because execution is automated. \\
\addlinespace
Certification may be noisy or strategic
& Provide procedural safeguards
& Standards should be published; refusals should be logged; agents should have appeal and recertification channels. \\
\addlinespace
Risk varies across environments
& Use proportionality rather than a universal rule
& Low-risk information tasks and high-stakes account operations should face different compliance burdens. \\
\bottomrule
\end{tabularx}
\end{table}

The table is a translation device rather than an additional formal result. It shows how the model's residual-control logic can discipline institutional design without implying that any single certification architecture is uniquely optimal.

\subsection{Certification Architecture}
\label{subsec:policy-certification-architecture}

A conditional delegation regime requires an institution capable of evaluating compliance. The model does not determine who should perform certification. It does, however, imply that the certifier's design matters because certification reallocates residual control over the account interface. A certification decision does not merely label a proxy as safe or unsafe; it determines whether the platform's right of refusal is limited.

There are three possible certification architectures.

A first approach is platform-administered certification. Platforms have detailed knowledge of their own systems, including authentication rules, rate limits, abuse patterns, security risks, and infrastructure constraints. This information advantage makes platform review technically efficient. The drawback is conflict of interest. A platform may use safety requirements to raise rivals' costs, protect its own assistant products, preserve advertising interfaces, or maintain existing monetization channels.

A second approach is government-administered certification. Public certification can reduce platform self-preferencing and provide stronger legal authority. Its weakness is speed and technical specificity. AI-agent capabilities, interface designs, security vulnerabilities, and abuse patterns change quickly. A centralized public process may update too slowly and may lack the operational detail needed for day-to-day interface governance.

A third approach is independent technical certification under public oversight. Independent auditors would assess agents against published requirements. Public authorities would set procedural safeguards: non-discrimination, transparency, appeal rights, periodic review, conflict-of-interest limits, and rules against certification capture. This approach is not costless, but it better separates technical evaluation from platform incentives.

The institutional choice should depend on the risk environment. Platform-administered review may be appropriate for low-risk integrations where the platform's informational advantage is large and exclusion incentives are weak. Independent certification under public oversight becomes more attractive when the platform competes with agent providers, when refusal can foreclose user-side intermediation, or when certified access has market-wide significance.

\subsection{Operationalizing the Compliance Frontier}
\label{subsec:policy-certification-dimensions}

In the model, certification is represented by a scalar threshold $\bar{K}$ on expected risk. In practice, this threshold must be translated into observable and auditable requirements. Certification should not evaluate only whether an agent can complete a task. It should evaluate whether the agent completes the task within authorized limits and leaves a reviewable record.

Table~\ref{tab:policy-certification-dimensions} lists the main dimensions of the certification interface.

\begin{table}[htbp]
\centering
\small
\caption{Certification Requirements}
\label{tab:policy-certification-dimensions}
\begin{tabularx}{\textwidth}{p{0.25\textwidth} X}
\toprule
\textbf{Dimension} & \textbf{Verifiable Requirement} \\
\midrule
Privacy and Data Minimization
& The proxy should access only the data needed for the authorized task; limit retention and onward transfer; prevent unauthorized telemetry; and delete relevant credentials or data after user revocation. \\
\addlinespace
System Load and Rate Limits
& The proxy should comply with rate limits, concurrency limits, retry rules, and back-off protocols designed to prevent congestion, scraping-like overload, or service degradation. \\
\addlinespace
Authorization Scope and Identity Protection
& The proxy should verify explicit user authorization, operate within the user's account rights, preserve the user's identity as principal, and disclose automated execution when disclosure is relevant to counterparties or platform integrity. \\
\addlinespace
Auditability and Logging
& The proxy should maintain reliable logs of authorization states, access events, execution timestamps, user confirmations, exceptional outcomes, and system responses so that disputes can be reviewed ex post. \\
\addlinespace
Security and Incident Response
& The proxy provider should maintain anomaly detection, suspension protocols, reporting channels, and procedures for credential compromise, prompt-injection attacks, unauthorized execution, or data leakage. \\
\addlinespace
Accountability and Liability Readiness
& The proxy provider should accept responsibility for out-of-scope execution, misleading compliance claims, failure to maintain required logs, and unsafe automation under its control. \\
\bottomrule
\end{tabularx}
\end{table}

These requirements make the abstract risk threshold operational. They also preserve the conceptual boundaries of delegation rights. Delegation is not account transfer, independent scraping, or unrestricted API access. It is authorized, revocable, scope-limited proxy execution of account privileges that the user already holds.

Certification should also be task-specific. Permission to compare product listings should not automatically authorize payment execution. Permission to summarize messages should not automatically authorize message transmission. Permission to organize travel options should not automatically authorize purchases, cancellations, or legally binding submissions. The standard should become stricter as the proxy moves from information processing to action execution, and stricter again when the action is irreversible, financially material, legally significant, or harmful to third parties if performed incorrectly.

\subsection{Dynamic Standards and Evidence-Based Refusal}
\label{subsec:policy-dynamic-refusal}

Certification standards should change as information changes. Let the risk threshold at time $t$ be written as
\begin{equation}
\bar{K}_t=\bar{K}(\mathcal{I}_t),
\label{eq:policy-dynamic-threshold}
\end{equation}
where $\mathcal{I}_t$ includes information about known vulnerabilities, realized load shocks, compliance costs, proxy failures, new agent capabilities, abuse patterns, third-party complaints, and observed patterns of platform refusal.

In practice, updating should combine periodic review with event-driven intervention. Periodic review allows standards to adjust as agent capabilities and compliance technologies improve. Event-driven updates allow a regulator, certifier, or platform to respond to new vulnerabilities, large-scale failures, credential compromise, privacy breaches, or evidence that a platform is using safety claims as a pretext for exclusion.

Certified Delegation does not eliminate the platform's right to protect its systems. A platform should be able to suspend or terminate proxy access when there is documented non-compliance, such as token expiration, user revocation, rate-limit violations, fraudulent injection, abnormal load, unauthorized data access, or verified privacy breaches. But because certification limits the platform's residual right of refusal, suspension should be governed by procedural safeguards.

The guiding principle is proportionality. Refusal should be evidence-based, risk-based, non-discriminatory, and reviewable. Emergency suspensions should be time-bounded and supported by logged technical reasons. When the risk is immediate and severe, temporary suspension may be justified before full review. When the risk is disputed or remediable, the certified agent should have access to notice, cure, appeal, and recertification procedures.

\subsection{User Authorization and Revocability}
\label{subsec:policy-user-authorization}

Technical certification of the proxy is necessary but not sufficient. Delegation is valid only when it rests on continuing authorization by the account holder. That authorization should be explicit, granular, task-specific, revocable, and limited to the user's existing account rights.

Explicit authorization means that the user knows which proxy is being authorized, which data fields are exposed, how long access lasts, which actions are permitted, and which actions require separate confirmation. Granular authorization means that the user can authorize one task without authorizing another. Task-specific authorization prevents permission for one workflow from becoming a general license to act inside the account.

Revocability is equally important. The user should be able to terminate the proxy's authority without terminating the underlying account relationship. Revocation should invalidate active credentials, stop pending automated execution, and trigger data-deletion or retention-limit obligations where applicable. A delegation regime without effective revocation would begin to resemble account transfer or uncontrolled third-party access rather than user-authorized proxy execution.

High-risk actions should require additional confirmation. These include irreversible financial transfers, changes to security settings, legally binding submissions, deletion of important data, account recovery actions, medical or legal communications, and messages or transactions that materially affect third parties. Users should also have access to logs showing what the proxy did, when it acted, under what authorization, and whether any action required additional confirmation.

These requirements preserve the defining features of delegation. The user remains the principal. The proxy acts within a limited scope. The platform retains evidence-based refusal rights against unsafe or non-compliant automation. Third parties receive protection through disclosure, logging, and accountability rules when proxy execution affects their interests.

\subsection{Liability and Risk Sharing}
\label{subsec:policy-liability}

The model does not derive a unique liability rule. It does, however, suggest a useful allocation principle: responsibility should fall on the party best positioned to prevent, monitor, or control the relevant risk. Liability rules should complement certification by making responsibility traceable when failures occur.

Table~\ref{tab:policy-liability-allocation} summarizes a possible allocation of primary responsibility.

\begin{table}[htbp]
\centering
\small
\caption{Liability Principles under Certified Delegation}
\label{tab:policy-liability-allocation}
\begin{tabularx}{\textwidth}{p{0.22\textwidth} X}
\toprule
\textbf{Party} & \textbf{Primary Responsibility} \\
\midrule
Agent Provider
& Responsible for out-of-scope execution, unsafe automation under its control, privacy breaches caused by the proxy, failure to maintain required logs, misleading compliance claims, and failure to respect authorization, revocation, or rate limits. \\
\addlinespace
Host Platform
& Responsible for arbitrary or discriminatory exclusion of certified proxies, self-preferencing disguised as security enforcement, failure to honor valid delegation credentials, failure to provide logged reasons for refusal, and negligence in baseline authentication or interface security. \\
\addlinespace
Account Holder
& Responsible for intentionally authorizing malicious proxies, providing fraudulent instructions, misusing a compliant proxy to carry out unlawful actions, or knowingly overriding required warnings and confirmations. \\
\addlinespace
Independent Certifier
& Responsible for grossly negligent review, corrupt certification, unmanaged conflicts of interest, failure to update standards after known risks, or failure to follow published certification procedures. \\
\bottomrule
\end{tabularx}
\end{table}

These rules can be supported by insurance, compliance bonds, escrow arrangements, logged technical notices, and fast dispute-resolution procedures. High-risk proxy providers may be required to maintain insurance or post a bond before certification. Platforms, in turn, should provide logged technical reasons when rejecting, suspending, or terminating a certified proxy. The purpose is not to eliminate all disputes. It is to make responsibility traceable and to reduce the ability of either side to exploit uncertainty strategically.

\subsection{Design Risks and Institutional Safeguards}
\label{subsec:policy-design-risks}

A certification-contingent access regime also creates design risks. Certification can become too lenient, allowing unsafe automation to receive protected access. It can become too strict, turning compliance into an entry barrier. It can become captured by platforms, certifiers, or incumbent agent providers. It can also become stale if standards fail to update as agent capabilities and threat models change.

These risks do not overturn the conditional-access principle, but they affect institutional design. A workable regime should include published standards, non-discriminatory review, proportional requirements, audit trails, appeal rights, periodic updating, and emergency suspension procedures. The certifier should distinguish between refusal based on documented technical risk and refusal based on a platform's preference to avoid user-side intermediation.

The regime should also account for third-party interests. In commerce environments, many delegation tasks primarily affect the delegating user and the host platform. In social, financial, medical, or identity-sensitive environments, proxy execution may affect non-delegating users, counterparties, or bystanders. Stronger standards are justified when delegated execution touches communications, social graphs, sensitive personal data, payment credentials, legal obligations, or identity signals.

\subsection{Conclusion}
\label{subsec:policy-conclusion}

This paper has argued that AI agents create a new control problem in digital platforms. The issue is not simply data ownership, API access, platform infrastructure, or account transferability. It is the mode through which an account holder exercises rights that she already has. We call this margin a delegation right.

The model shows why this margin matters. If the platform holds an unconditional veto over automated account use, users and agent providers may underinvest in workflows, adaptation, and compliance because access can be withheld after investments are sunk. If the user holds an unconditional right to delegate, the hold-up problem is reduced, but some risks to the platform and to third parties may remain outside the agent's private incentives.

Certified Delegation offers a conditional alternative. It protects user-authorized delegation when the proxy satisfies verifiable standards of authorization, revocability, auditability, risk control, and accountability. At the same time, it preserves the platform's ability to exclude uncertified or non-compliant automation. Certification is therefore not only a safety screen; it is a conditional allocation of residual control over the account interface.

The illustrative mechanism simulations show how such a rule can reduce deadweight loss by restoring delegation-related investment while limiting residual risk. The scale-normalized counterfactuals further suggest that the appropriate standard should vary across environments, with lighter requirements for low-risk information tasks and stricter requirements for privacy-sensitive, identity-sensitive, or high-stakes account operations. These numerical exercises are not estimates of the welfare effects of actual disputes. They are illustrations of the model's mechanisms under maintained parameter values.

As AI agents become a more common interface for digital consumption and exchange, platform governance will need rules more precise than either platform absolutism or unconditional access mandates. A workable regime should make delegation conditional, revocable, auditable, evidence-based, and proportionate to risk. The paper's central claim is therefore limited but important: a user-authorized, scope-limited, identity-preserving, and accountable proxy should receive protection against exclusion when it satisfies a verifiable certification standard, while unsafe or non-compliant automation should remain subject to platform refusal.
\appendix

=======================================
\bibliographystyle{plainnat}
\bibliography{references}

\begin{thebibliography}{45}
\providecommand{\natexlab}[1]{#1}
\providecommand{\url}[1]{\texttt{#1}}
\expandafter\ifx\csname urlstyle\endcsname\relax
  \providecommand{\doi}[1]{doi: #1}\else
  \providecommand{\doi}{doi: \begingroup \urlstyle{rm}\Url}\fi

\bibitem[Aghion and Tirole(1997)]{AghionTirole1997}
Philippe Aghion and Jean Tirole.
\newblock Formal and real authority in organizations.
\newblock \emph{Journal of Political Economy}, 105\penalty0 (1):\penalty0 1--29, 1997.

\bibitem[{Anthropic}(2024)]{AnthropicComputerUse2024}
{Anthropic}.
\newblock Introducing computer use, a new claude 3.5 sonnet, and claude 3.5 haiku.
\newblock Anthropic announcement, 2024.

\bibitem[Armstrong(2006)]{Armstrong2006}
Mark Armstrong.
\newblock Competition in two-sided markets.
\newblock \emph{RAND Journal of Economics}, 37\penalty0 (3):\penalty0 668--691, 2006.

\bibitem[Boudreau(2010)]{Boudreau2010}
Kevin~J. Boudreau.
\newblock Open platform strategies and innovation: Granting access vs. devolving control.
\newblock \emph{Management Science}, 56\penalty0 (10):\penalty0 1849--1872, 2010.

\bibitem[Boudreau and Hagiu(2009)]{BoudreauHagiu2009}
Kevin~J. Boudreau and Andrei Hagiu.
\newblock Platform rules: Multi-sided platforms as regulators.
\newblock Working paper, Harvard Business School, 2009.

\bibitem[Calo(2015)]{Calo2015}
Ryan Calo.
\newblock Robotics and the lessons of cyberlaw.
\newblock \emph{California Law Review}, 103\penalty0 (3):\penalty0 513--563, 2015.

\bibitem[Cr\'emer et~al.(2019)Cr\'emer, de~Montjoye, and Schweitzer]{CremerMontjoyeSchweitzer2019}
Jacques Cr\'emer, Yves-Alexandre de~Montjoye, and Heike Schweitzer.
\newblock Competition policy for the digital era.
\newblock Technical report, European Commission, 2019.

\bibitem[Eisenmann et~al.(2009)Eisenmann, Parker, and Van~Alstyne]{EisenmannParkerVanAlstyne2009}
Thomas Eisenmann, Geoffrey Parker, and Marshall Van~Alstyne.
\newblock Opening platforms: How, when and why?
\newblock In Annabelle Gawer, editor, \emph{Platforms, Markets and Innovation}, pages 131--162. Edward Elgar, 2009.

\bibitem[{European Parliament and Council}(2016)]{GDPR2016}
{European Parliament and Council}.
\newblock Regulation (eu) 2016/679 (general data protection regulation), article 20: Right to data portability.
\newblock Official Journal of the European Union, L 119, 2016.

\bibitem[{European Parliament and Council}(2022{\natexlab{a}})]{DMA2022}
{European Parliament and Council}.
\newblock Regulation (eu) 2022/1925 on contestable and fair markets in the digital sector (digital markets act).
\newblock Official Journal of the European Union, L 265, 2022{\natexlab{a}}.

\bibitem[{European Parliament and Council}(2022{\natexlab{b}})]{DSA2022}
{European Parliament and Council}.
\newblock Regulation (eu) 2022/2065 on a single market for digital services (digital services act).
\newblock Official Journal of the European Union, L 277, 2022{\natexlab{b}}.

\bibitem[{European Parliament and Council}(2023)]{DataAct2023}
{European Parliament and Council}.
\newblock Regulation (eu) 2023/2854 on harmonised rules on fair access to and use of data (data act).
\newblock Official Journal of the European Union, L 2023/2854, 2023.

\bibitem[{European Parliament and Council}(2024)]{EUAIAct2024}
{European Parliament and Council}.
\newblock Regulation (eu) 2024/1689 laying down harmonised rules on artificial intelligence (ai act).
\newblock Official Journal of the European Union, 2024.

\bibitem[Graef et~al.(2013)Graef, Verschakelen, and Valcke]{GraefVerschakelenValcke2013}
Inge Graef, Jeroen Verschakelen, and Peggy Valcke.
\newblock Putting the right to data portability into a competition law perspective.
\newblock \emph{Law: The Journal of the Higher School of Economics}, pages 53--63, 2013.
\newblock No. 3.

\bibitem[Grossman and Hart(1986)]{GrossmanHart1986}
Sanford~J. Grossman and Oliver~D. Hart.
\newblock The costs and benefits of ownership: A theory of vertical and lateral integration.
\newblock \emph{Journal of Political Economy}, 94\penalty0 (4):\penalty0 691--719, 1986.

\bibitem[Hardt(2012)]{RFC6749}
D.~Hardt.
\newblock The {OAuth} 2.0 authorization framework.
\newblock IETF RFC 6749, 2012.

\bibitem[Hart(1995)]{Hart1995}
Oliver Hart.
\newblock \emph{Firms, Contracts, and Financial Structure}.
\newblock Oxford University Press, 1995.

\bibitem[Hart and Moore(1990)]{HartMoore1990}
Oliver Hart and John Moore.
\newblock Property rights and the nature of the firm.
\newblock \emph{Journal of Political Economy}, 98\penalty0 (6):\penalty0 1119--1158, 1990.

\bibitem[Holmstr\"om and Tirole(1989)]{HolmstromTirole1989}
Bengt Holmstr\"om and Jean Tirole.
\newblock The theory of the firm.
\newblock In Richard Schmalensee and Robert~D. Willig, editors, \emph{Handbook of Industrial Organization}, volume~1, pages 61--133. Elsevier, 1989.

\bibitem[{International Organization for Standardization}(2023)]{ISO42001}
{International Organization for Standardization}.
\newblock {ISO/IEC} 42001:2023 --- artificial intelligence management system.
\newblock ISO, 2023.

\bibitem[Klein et~al.(1978)Klein, Crawford, and Alchian]{KleinCrawfordAlchian1978}
Benjamin Klein, Robert~G. Crawford, and Armen~A. Alchian.
\newblock Vertical integration, appropriable rents, and the competitive contracting process.
\newblock \emph{Journal of Law and Economics}, 21\penalty0 (2):\penalty0 297--326, 1978.
\newblock \doi{10.1086/466922}.

\bibitem[Kr\"amer(2021)]{Kramer2021}
Jan Kr\"amer.
\newblock Personal data portability in the platform economy: Economic implications and policy recommendations.
\newblock \emph{Journal of Competition Law \& Economics}, 17\penalty0 (2):\penalty0 263--308, 2021.
\newblock \doi{10.1093/joclec/nhaa030}.

\bibitem[Kr\"amer and St\"udlein(2019)]{KramerStuedlein2019}
Jan Kr\"amer and Nadine St\"udlein.
\newblock Data portability, data disclosure and data-induced switching costs: Some unintended consequences of the general data protection regulation.
\newblock \emph{Economics Letters}, 181:\penalty0 99--103, 2019.
\newblock \doi{10.1016/j.econlet.2019.05.015}.

\bibitem[Kr\"amer et~al.(2020)Kr\"amer, Senellart, and de~Streel]{KramerSenellartDeStreel2020}
Jan Kr\"amer, Pierre Senellart, and Alexandre de~Streel.
\newblock Making data portability more effective for the digital economy.
\newblock Technical report, Centre on Regulation in Europe (CERRE), 2020.

\bibitem[Lodderstedt et~al.(2023)Lodderstedt, Richer, and Campbell]{RFC9396}
T.~Lodderstedt, J.~Richer, and B.~Campbell.
\newblock {OAuth} 2.0 rich authorization requests.
\newblock IETF RFC 9396, 2023.

\bibitem[Lodderstedt et~al.(2025)Lodderstedt, Bradley, Labunets, and Fett]{RFC9700}
T.~Lodderstedt, J.~Bradley, A.~Labunets, and D.~Fett.
\newblock Best current practice for {OAuth} 2.0 security.
\newblock IETF RFC 9700, 2025.

\bibitem[{National Institute of Standards and Technology}(2023)]{NISTAIRMF2023}
{National Institute of Standards and Technology}.
\newblock Ai risk management framework (ai rmf 1.0).
\newblock NIST, 2023.

\bibitem[{National Institute of Standards and Technology}(2024)]{NISTGenAIProfile2024}
{National Institute of Standards and Technology}.
\newblock Artificial intelligence risk management framework: Generative artificial intelligence profile.
\newblock NIST AI 600-1, 2024.

\bibitem[{OECD}(2021)]{OECD2021Portability}
{OECD}.
\newblock Data portability, interoperability and competition.
\newblock Technical report, OECD Directorate for Financial and Enterprise Affairs, Competition Committee, 2021.

\bibitem[{OpenAI}(2025{\natexlab{a}})]{OpenAIChatGPTAgent2025}
{OpenAI}.
\newblock Chatgpt agent system card.
\newblock OpenAI, 2025{\natexlab{a}}.

\bibitem[{OpenAI}(2025{\natexlab{b}})]{OpenAIOperator2025}
{OpenAI}.
\newblock Operator system card.
\newblock OpenAI, 2025{\natexlab{b}}.

\bibitem[Parker and Van~Alstyne(2005)]{ParkerVanAlstyne2005}
Geoffrey~G. Parker and Marshall~W. Van~Alstyne.
\newblock Two-sided network effects: A theory of information product design.
\newblock \emph{Management Science}, 51\penalty0 (10):\penalty0 1494--1504, 2005.

\bibitem[Raji et~al.(2020)Raji, Smart, White, Mitchell, Gebru, Hutchinson, Smith-Loud, Theron, and Barnes]{RajiEtAl2020}
Inioluwa~Deborah Raji, Andrew Smart, Rebecca~N. White, Margaret Mitchell, Timnit Gebru, Ben Hutchinson, Jamila Smith-Loud, Daniel Theron, and Parker Barnes.
\newblock Closing the {AI} accountability gap: Defining an end-to-end framework for internal algorithmic auditing.
\newblock In \emph{Proceedings of the 2020 Conference on Fairness, Accountability, and Transparency}, pages 33--44, 2020.
\newblock \doi{10.1145/3351095.3372873}.

\bibitem[Rochet and Tirole(2003)]{RochetTirole2003}
Jean-Charles Rochet and Jean Tirole.
\newblock Platform competition in two-sided markets.
\newblock \emph{Journal of the European Economic Association}, 1\penalty0 (4):\penalty0 990--1029, 2003.

\bibitem[Rochet and Tirole(2006)]{RochetTirole2006}
Jean-Charles Rochet and Jean Tirole.
\newblock Two-sided markets: A progress report.
\newblock \emph{RAND Journal of Economics}, 37\penalty0 (3):\penalty0 645--667, 2006.

\bibitem[Shapley(1953)]{Shapley1953}
Lloyd~S. Shapley.
\newblock A value for n-person games.
\newblock In Harold~W. Kuhn and Albert~W. Tucker, editors, \emph{Contributions to the Theory of Games II}, pages 307--317. Princeton University Press, 1953.

\bibitem[{Stigler Committee on Digital Platforms}(2019)]{StiglerReport2019}
{Stigler Committee on Digital Platforms}.
\newblock Final report.
\newblock Technical report, Stigler Center, University of Chicago, 2019.

\bibitem[{Supreme Court of the United States}(2021)]{VanBuren2021}
{Supreme Court of the United States}.
\newblock Van buren v. united states, 593 u.s. 374 (2021), 2021.

\bibitem[Swire and Lagos(2013)]{SwireLagos2013}
Peter Swire and Yianni Lagos.
\newblock Why the right to data portability likely reduces consumer welfare: Antitrust and privacy critique.
\newblock \emph{Maryland Law Review}, 72\penalty0 (2):\penalty0 335--380, 2013.

\bibitem[Tirole(1999)]{Tirole1999}
Jean Tirole.
\newblock Incomplete contracts: Where do we stand?
\newblock \emph{Econometrica}, 67\penalty0 (4):\penalty0 741--781, 1999.

\bibitem[Tiwana et~al.(2010)Tiwana, Konsynski, and Bush]{TiwanaKonsynskiBush2010}
Amrit Tiwana, Benn Konsynski, and Ashley~A. Bush.
\newblock Platform evolution: Coevolution of platform architecture, governance, and environmental dynamics.
\newblock \emph{Information Systems Research}, 21\penalty0 (4):\penalty0 675--687, 2010.

\bibitem[{United States Court of Appeals for the Ninth Circuit}(2022)]{hiQLinkedIn2022}
{United States Court of Appeals for the Ninth Circuit}.
\newblock hiq labs, inc. v. linkedin corp., 31 f.4th 1180 (9th cir.\ 2022), 2022.

\bibitem[Weitzenb\"ock(2004)]{Weitzenboeck2004}
Emily~M. Weitzenb\"ock.
\newblock Good faith and fair dealing in contracts formed and performed by electronic agents.
\newblock \emph{Artificial Intelligence and Law}, 12\penalty0 (1--2):\penalty0 83--110, 2004.

\bibitem[Williamson(1985)]{Williamson1985}
Oliver~E. Williamson.
\newblock \emph{The Economic Institutions of Capitalism: Firms, Markets, Relational Contracting}.
\newblock Free Press, New York, 1985.

\bibitem[Zhou et~al.(2023)Zhou, Xu, Zhu, Zhou, Lo, Sridhar, Cheng, Ou, Bisk, Fried, Alon, and Neubig]{WebArena2023}
Shuyan Zhou, Frank~F. Xu, Hao Zhu, Xuhui Zhou, Robert Lo, Abishek Sridhar, Xianyi Cheng, Tianyue Ou, Yonatan Bisk, Daniel Fried, Uri Alon, and Graham Neubig.
\newblock {WebArena}: A realistic web environment for building autonomous agents.
\newblock arXiv:2307.13854, 2023.

\end{thebibliography}

\end{document}